\documentclass[a4paper,12pt]{amsart}
\usepackage[left=1.2in, right=1.2in, , bottom = 1.2in, top = 1.2in]{geometry}
\usepackage[utf8]{inputenc}
\usepackage[all]{xy}
\usepackage{amssymb}
\usepackage{amsmath}
\usepackage{csquotes}
\usepackage{subfig}
\usepackage{enumerate}
\usepackage{graphicx}
\usepackage{enumerate}
\usepackage{amsaddr}
\usepackage{color}
\usepackage{mwe}
\usepackage{subfig}
\usepackage{appendix}

\usepackage{soul}
\usepackage{xcolor}
\newcommand\independent{\protect\mathpalette{\protect\independenT}{\perp}}
\def\independenT#1#2{\mathrel{\rlap{$#1#2$}\mkern2mu{#1#2}}}

\usepackage{tikz}
\usetikzlibrary{arrows.meta,shapes}
\usetikzlibrary{arrows,shapes.arrows,shapes.geometric,shapes.multipart,
decorations.pathmorphing,positioning,shapes.swigs,}
\usepackage{setspace}
\newtheorem{lemma}{Lemma}

\newtheorem{theorem}{Theorem}

\MakeOuterQuote{"}\EnableQuotes
\DeclareUnicodeCharacter{00A0}{ }

\title[]{Conditional separable effects}
\author{Mats J. Stensrud$^{1}$,James M. Robins$^{2}$, Aaron Sarvet$^{2}$, Eric J. Tchetgen Tchetgen$^{3}$,   Jessica G. Young$^{2,4}$} \address{ \small $^1$ Department of Mathematics, Ecole Polytechnique Fédérale de Lausanne, Switzerland \\
$^2$ Department of Epidemiology, Harvard T. H. Chan School of Public Health, USA  \\
$^3$ Department of Statistics, The Wharton School, University of Pennsylvania, USA \\
$^4$ Department of Population Medicine,
Harvard Medical School, USA \\
}
\begin{document}
\maketitle

\begin{abstract}
Researchers are often interested in treatment effects on outcomes that are only defined conditional on a post-treatment event status. For example, in a study of the effect of different cancer treatments on quality of life at end of follow-up, the quality of life of individuals who die during the study is undefined. In these settings, a naive contrast of outcomes conditional on the post-treatment variable is not an average causal effect, even in a randomized experiment. Therefore the effect in the principal stratum of those who would have the same value of the post-treatment variable regardless of treatment, such as the always survivors in a truncation by death setting, is often advocated for causal inference. While this principal stratum effect is a well defined causal contrast, it is often hard to justify that it is relevant to scientists, patients or policy makers, and it cannot be identified without relying on unfalsifiable assumptions. Here we formulate alternative estimands, the \emph{conditional separable effects}, that have a natural causal interpretation under assumptions that can be falsified in a randomized experiment. A feature of the conditional separable effects is that the investigator must describe modified versions of the original  treatment, motivated by subject-matter knowledge. We provide identification results and introduce different estimators, including a doubly robust estimator derived from the nonparametric influence function. As an illustration, we estimate a conditional separable effect of chemotherapies on quality of life in patients with prostate cancer, using data from a randomized clinical trial.
\end{abstract}

\section{Introduction}
\label{sec: intro}
Many research questions involve treatment effects on outcomes that are only defined conditional on a post-treatment event status. For example, in a study of the effects of different cancer treatments on quality of life at end of follow-up, the quality of life of individuals who die during the study is undefined. Furthermore, some treatment effects are only of substantive interest conditional on a post-treatment event. In a study of the effect of a vaccine on viral load at end of follow-up, the viral load of individuals who never become infected is not of substantive interest (even though it is defined).

By design of a randomized trial, we can identify a counterfactual contrast of mean outcomes across treatment arms conditional on a post-treatment status (when there are no losses to follow-up). However, this contrast does not in general equal a causal effect when the treatment affects the post-treatment event.  In this case, outcomes are being compared in different sets of individuals, and thus the comparison cannot be interpreted as a contrast of (counterfactual) outcomes in the same set of individuals under different treatment conditions.  For example, when the cancer treatments affect survival, the subset of the population who would survive under one treatment will be different from the subset who would survive under the other treatment.  

Selecting a meaningful definition of a causal effect in this setting is not straightforward. One option is to consider a so-called \textsl{controlled direct effect} \cite{robins1992identifiability}, which quantifies the effect of the treatment on the outcome had we (somehow) eliminated post-treatment events that render the outcome undefined or not of substantive interest.  However, such effects often do not quantify effects of interest to an investigator, policy maker, doctor or patient; for example, the utility of the effect of a cancer treatment on quality of life had we eliminated death, or the effect of a vaccine on viral load had we forced everyone to become infected, is not clear.

Given the limitations of the controlled direct effect, Robins \cite{robins1986new} introduced the \textsl{principal stratum effects} in settings where the outcome of interest is only defined, or of substantive interest, conditional on a post-treatment event status; in particular,  the causal effect in the subset of individuals who would have this event status, regardless of the treatment they were given. The name "principal stratum effect" is due to Frangakis and Rubin \cite{frangakis2002principal} who, in contrast to Robins's \cite{robins1986new} more skeptical view, advocated strongly for the use of this estimand. Indeed, it has been argued that no other sensible causal estimand exists in this case \cite{ding2017principal, vanderweele2011principal, rubin2006causal}. However, principal stratum effects also have several serious limitations \cite{robins1986new, robins2007principal, joffe2011principal, dawid2012imagine, robins2010alternative,robins2020interventionist}. It is impossible to observe who comprises this subset of the population because it is defined by the event status in the same individual under different treatments.  Further, this subset may not actually exist and, if it does exist, may constitute a highly unusual subset of the original population.  Finally, identification of a principal stratum effect generally relies on assumptions that cannot be falsified in any real-world experiment. 

Here we provide new definitions of causal effects for settings where the outcome is only defined or of substantive interest conditional on a particular post-treatment event status: the \emph{conditional separable effects}.  We give general conditions under which these effects can be identified along with various estimators.  These estimands are inspired by Robins and Richardson's treatment decomposition in the context of mediation \cite{robins2010alternative} and the (marginal) separable effects for competing events settings \cite{stensrud2019separable, stensrud2019generalized}. However, the problem of identifying the conditional separable effects is non-trivially distinct from the problem of identifying their marginal counterparts, leading to distinct identifying functionals from those considered in this previous work.  In turn, the estimation problem here is distinct and the estimators we develop do not overlap with those developed in this previous work.

Unlike principal stratum effects, the conditional separable effects rely only on assumptions that can be falsified in future real-world experiments.  Among these is a simple isolation condition which ensures that (i) the conditional separable effects quantify meaningful direct effects of the treatment on the outcome in a particular subset of the population and (ii) the individuals comprising this subset can be identified. 

Regardless of whether the isolation condition holds, we argue that critical thinking about this condition is essential when the outcome is undefined or not of substantive interest conditional on a post-treatment event status.  In particular, we argue that, when it fails, it is not clear that a meaningful notion of this conditional effect exists. Finally, we show that, given isolation conditions, a conditional separable effect equals a principal stratum effect under an additional monotonicity condition. 




The manuscript is organized as follows. In Section \ref{sec: contribution} we clarify the novel contribution of the manuscript. In Section \ref{sec: observed data structure} we describe the observed data structure. In Section \ref{sec: principal stratum}, we review the interpretation and identification of conventional principal stratum effects. In Section \ref{sec: def conditional separable effects}, we define the conditional separable effects and discuss their interpretation under the isolation condition. In Section \ref{sec: modif treat}, we define a modified treatment assumption \cite{stensrud2019generalized}, allowing the conditional separable effects to explain mechanism of the original treatment under study. In Section \ref{sec: Ay partial meaningful eff} we discuss the need for the isolation condition and modified treatment assumption in order to define a meaningful effect of the original treatment on the outcome in this setting and give an example. In Section \ref{sec: identifiability conditions}, we define conditions that are sufficient to identify the conditional separable effects with the observed data, provide an identification formula, and give new identification results for classical principal stratum estimands in the presence of time-varying common causes of the outcome of interest and the conditioning event. In Section \ref{sec: estimation}, we describe different estimators of the conditional separable effect, including a doubly robust estimator based on the nonparametric influence function. In Section \ref{sec: data example}, we apply our results to understand the effect of different chemotherapies on quality of life in patients with prostate cancer. In Section \ref{sec: discussion}, we end with a discussion.

\section{Contribution of this work}
\label{sec: contribution}
We have previously considered \textit{marginal} separable effects in competing events settings \cite{stensrud2019separable, stensrud2019generalized}. Here, we introduce new causal estimands that are defined \textit{conditional} on a post-treatment event, and that generalize beyond time-to-event outcomes. The conditions that we define for identification of the conditional separable effects lead to distinct identifying functionals from those for the marginal separable effects.  Further, we show that these identifying functionals identify principal stratum effects under the additional assumption of monotonicity, extending the sparse existing literature on principal stratum effects in the presence of time-varying covariates \cite{tchetgen2014identification}. Given the distinct identifying functionals, the estimators we develop here, which include a doubly robust estimator, do not overlap with those developed in the work by Stensrud et al \cite{stensrud2019separable, stensrud2019generalized}.  

\section{Observed Data}
\label{sec: observed data structure}
Consider a randomized experiment with $n$ i.i.d.\ individuals who are assigned a binary treatment $A\in \{0,1\}$ at baseline. Let $k\in \{0,\dots,K+1\}$ index equally spaced discrete time intervals, and $Y\equiv Y_{K+1}$ is the outcome of interest measured in $K+1$. Let $D_{k+1}$ be an indicator of post-treatment event status by $k+1$, such that $Y$ is only defined, or is only of substantial interest, when $D_{k+1}=0$, $k\in \{0,\dots,K\}$. For example, in our running example on cancer treatment and quality of life, $D_{k+1}$ is an indicator of death by $k+1$.  In our other example on vaccination and viral load, $D_{k+1}$ is an indicator of not becoming infected by $k+1$.


Let $L_{0}$ denote a vector of pre-randomization (baseline) covariates, and $L_{k}$ a vector of post-randomization (time-varying) covariates measured at $k$. We adopt the temporal
convention $(L_{0},A,\dots,D_{k},L_{k},\dots,D_{K+1},Y)$, and this assumption does not restrict the distribution of observed data when we let the interval length become infintesimally small.  We use overbars (e.g. $\overline{D}_{k+1}$) to denote the history and underbars (e.g. $\underline{D}_{k+1}$) to denote the future of a random variable relative to $k+1$, respectively. We assume that no subject is lost to follow-up throughout the main text but extend all results to settings with loss to follow-up (censoring) in Appendix \ref{app sec: proof of id}.

\section{Challenges in interpretation and identification of principal stratum effects}
\label{sec: principal stratum}

For any individual in the study population, let $Y^{a}$ and $D^{a}_{k+1}$, $k\in \{0,\dots,K\}$, denote the counterfactual outcome of interest and the post-treatment event indicator by $k+1$, respectively, had she been assigned to $A=a$. The following additive \textsl{principal stratum effect}
\begin{align}
& \mathbb{E} ( Y^{a=1} - Y^{a=0} \mid 
D^{a=1}_{K+1}=D^{a=0}_{K+1}=0 )
\label{eq: SACE}
\end{align}
is the effect of the treatment $A$ on $Y$ in the subset of individuals who would never experience the post-treatment event rendering $Y$ undefined/not of interest under any level of treatment \cite{frangakis2002principal,robins1986new}.   
Because the conditioning set in \eqref{eq: SACE} is defined by outcomes under different counterfactual treatments, it is impossible to observe the individuals in this subset of the population \cite{robins1986new}, which clearly limits its practical relevance \cite{robins2007principal, joffe2011principal, dawid2012imagine}. Further, we are not guaranteed that this unknown subset of the population exists and, if it does exist, it may constitute a highly unusual subset of the original study population. 

Returning to our cancer treatment example, \eqref{eq: SACE} is the effect of cancer treatment on quality of life at end of follow-up among those who would survive throughout the study regardless of what cancer treatment they received.  This may be a highly unusual subgroup, particularly if $a=0$ refers to no treatment.  Similarly, in our vaccine example, \eqref{eq: SACE} is the effect of receiving the vaccine on viral load at end of follow-up among those who would become infected regardless of whether they received the vaccine.  For an effective vaccine, this subgroup may constitute a small and unusual segment of the original population (which, again, cannot be observed).

In addition to these interpretational challenges, strong assumptions are required for identification of principal stratum effects like \eqref{eq: SACE} \cite{ding2017principal,robins1986new,robins2007principal,hayden2005estimator, tchetgen2014identification}, even in idealized settings with a randomly assigned point treatment, no loss to follow-up and no common causes of $Y$ and $D_1$ as represented in the causal directed acyclic graph (DAG) \cite{pearl} in Figure \ref{fig: decomposition 1}a. Throughout, we will use causal DAGs to represent underlying assumptions on how random variables in a particular study are generated (including counterfactual random variables below). Specifically our causal DAGs represent an underlying Finest Fully Randomized Causally Interpreted Structural Tree Graph (FFRCISTG) model (as fine as the data) \cite{robins1986new,richardson2013single},  which is a counterfactual causal model that predates and makes fewer assumptions than the perhaps more familiar non-parametric structural equation model with independent errors (NPSEM-IE) \cite{robins2010alternative, pearl, richardson2013single}. The absence of an arrow on a causal DAG will be used to represent the assumption that an individual-level causal effect is absent for every subject.  

A causal DAG must minimally represent all common causes of any variable represented on the DAG.  Therefore, Figure \ref{fig: decomposition 1}a represents a generally restrictive assumption on the study data generating process because it depicts no common causes (measured or unmeasured) of $D_1$ and $Y$, which cannot be guaranteed even in a perfectly executed trial.  Only a handful of authors have considered identification of principal stratum effects when common causes of the conditioning event and the outcome of interest exist and may be affected by treatment \cite{tchetgen2014identification}; this more realistic data generating assumption is represented in Figure \ref{fig: isolation conditions}a.   Further, previously posed identification strategies for \eqref{eq: SACE} have relied on \textsl{unfalsifiable} assumptions; that is, assumptions that can never be challenged in any plausible experiment. 

In the sections that follow, we will introduce new counterfactual estimands that overcome these limitations of principal stratum effects.  These estimands can in principle be identified by design in a future real-world experiment.  Further, under a set of additional assumptions, that are falsifiable in such an experiment, these new estimands can be identified in a current experiment or observational study and capture a meaningful notion of treatment mechanism in an identified subset of the original study population. These assumptions are compatible with the existence of common causes (possibly affected by treatment) of the conditioning event and the outcome of interest. Finally, we show that under one additional assumption, these estimands happen to equal a principal stratum effect such as \eqref{eq: SACE}.



\section{The Conditional Separable Effects}
\label{sec: def conditional separable effects}

Following Robins et al \cite{robins2010alternative,robins2020interventionist} in a mediation context and Stensrud et al \cite{stensrud2019separable, stensrud2019generalized} in a competing event context, suppose a four-arm trial could be plausibly conducted such that, in place of assignment   to one of the two values of $A$ as in Section \ref{sec: observed data structure}, individuals are jointly assigned values of two new treatments $A_Y\in \{0,1\}$ and $A_D\in \{0,1\}$.   Let $Y^{a_Y,a_D}$ and $D_{k+1}^{a_Y,a_D}$ denote the counterfactual outcome of interest and the post-treatment event indicator by $k+1$, respectively, had an individual been assigned $A_Y=a_Y$ and $A_D=a_D$ for $a_Y$ and $a_D$ possible realizations of $A_Y$ and $A_D$, respectively. We denote this four-arm trial by $G$. Consider the following condition relative to a causal DAG representing the assumed data generating mechanism under $G$,
\begin{align}
    & \text{there are no causal paths from } A_Y(G) \text{ to } D_{k+1}(G), \quad k \in \{0,\dots,K\}, \label{def: Ay partial iso} 
\end{align}
where any node $V(G)$ represented on the DAG denotes a random variable $V$ under the treatment assignment in $G$. Like the causal DAG representing the data generating mechanism for the current study of Section \ref{sec: observed data structure} (e.g.\ Figure \ref{fig: decomposition 1}a), the causal DAG representing the mechanism for the future trial $G$ (e.g.\ Figure \ref{fig: decomposition 1}b) relies on subject matter expertise/assumptions relative to the modified treatments $A_Y$ and $A_D$.  

Condition \eqref{def: Ay partial iso} relates to the condition of $A_Y$ partial isolation for time-to-event outcomes given in Stensrud et al \cite{stensrud2019generalized}, and therefore we will use the same terminology here.\footnote{In the mediation context, Robins and Richardson \cite[Section 6.1]{robins2010alternative} provided (using different nomenclature) examples of $A_Y$ partial isolation (their Figure 6a) and the related conditions $A_D$ partial isolation (their Fig 6.b) and full isolation (their Figure 4) (see also Stensrud \cite{stensrud2019generalized} and Section \ref{sec: full iso section} in this article).  Robins and Richardson showed identification of the marginal distribution of $Y^{a_D,a_Y}$, for $a_D \ne a_Y$,under the assumptions that $A$ was randomly assigned and \eqref{assumption: Determinsm} held. However they showed the identifying formulas depended on whether the true causal DAG satisfied $A_Y$ versus $A_D$ partial isolation. In contrast to the current paper Robins and Richardson \cite[Section 6.1]{robins2010alternative} did not consider interpretation and identification of conditional effects.} However, the version of $A_Y$ partial isolation we consider here (i.e.\ condition  \eqref{def: Ay partial iso}) is defined relative to outcomes that may not be time-to-events. $A_Y$ partial isolation is illustrated in the causal DAGs of Figure \ref{fig: decomposition 1}b and Figures \ref{fig: isolation conditions}b-c (with the index $G$ suppressed in the latter and in subsequently presented causal DAGs under $G$ to avoid clutter). 
By contrast, $A_Y$ partial isolation fails in Figure \ref{fig: isolation conditions}d due to the path $ A_Y(G) \rightarrow Z_1(G) \rightarrow D_2(G)$.  In a causal DAG \cite{pearl, robins2010alternative}, $A_Y$ partial isolation \eqref{def: Ay partial iso} ensures that $D_{k+1}^{a_Y,a_D} = D_{k+1}^{a_D}, \forall a_Y, a_D \in \{0,1\}$ and $\forall k \in \{0,\dots, K\} $ where $D_{k+1}^{a_D}$ denotes the counterfactual value of the post-treatment event indicator by $k+1$, had an individual been assigned $A_D=a_D$ and the natural value of $A_Y$ \cite{richardson2013single}. In other words, $A_Y$ partial isolation ensures that the treatment $A_Y$ does not directly or indirectly affect the post-treatment event by any $k+1$.  Importantly, $A_Y$ partial isolation is a falsifiable assumption.  For example, we can observe in the future 4-arm trial $G$ whether $$E\left[D_{k+1}(G) \mid A_Y(G)=1, A_D(G)=a_D\right]= E\left[D_{k+1}(G) \mid A_Y(G)=0, A_D(G)=a_D\right],k=0,\dots,K, $$ and $a_D=\{0,1\}$, with failure of this equality falsifying $A_Y$ partial isolation. 

The relation between $A_Y$ partial isolation in a causal DAG representing an underlying FFRCISTG model and the equality $D_{k+1}^{a_Y,a_D} = D_{k+1}^{a_D}$ can be more explicitly seen by minimal counterfactual labelling in a Single World Intervention Graph (SWIG), which explicitly depicts counterfactual variables \cite{richardson2013single}.  For example, Figure \ref{fig: swig minimal labelling}a depicts a SWIG that is a transformation of the causal DAG in Figure \ref{fig: isolation conditions}c, consistent with $A_Y$ partial isolation, under an intervention that sets $A_Y$ to $a_Y$ and $A_D$ to $a_D$.  Minimal labelling allows removal of the $a_Y$ superscript on the counterfactual values under intervention of both $Z_1$ and $D_2$.  By contrast, minimal labelling does not allow removal of this $a_Y$ superscript in Figure \ref{fig: swig minimal labelling}b which is a corresponding transformation of the causal DAG in \ref{fig: isolation conditions}d consistent with failure of $A_Y$ partial isolation.

Under $A_Y$ partial isolation, the following counterfactual contrast 
\begin{align}
&\mathbb{E} ( Y^{a_Y=1,a_D}  - Y^{a_Y=0,a_D} \mid
D^{a_D}_{K+1}=0 )
\label{eq: conditional sep eff 1}
\end{align}
is the average causal effect of the treatment $A_Y$ on $Y$ when all individuals are assigned $A_D=a_D$, in the subset of individuals who do not experience the post-treatment event under $A_D=a_D$, regardless of the value of $A_Y$ they are assigned.  Thus, under $A_Y$ partial isolation this subset of individuals can be directly observed in a study that assigns $A_Y$ and $A_D$ as simply the subset of individuals who do not experience the post-treatment event among all those receiving $A_D=a_D$.  This is in contrast to the unobservable subset of individuals who define the principal stratum effect \eqref{eq: SACE}.  We refer to \eqref{eq: conditional sep eff 1} as the \textsl{conditional separable effect} evaluated at $a_D \in \{0,1\}$. We take the assumption of $A_Y$ partial isolation as given throughout the remainder of this manuscript, unless otherwise stated. In Section \ref{sec: Ay partial meaningful eff} we will argue that this assumption is required to meaningfully define an effect of $A$ on $Y$ conditional on $\overline{D}_{K+1}=0$.  

In Section \ref{sec: prostate revisited}, we discuss two future treatments $A_Y$ and $A_D$ consistent with the assumption of $A_Y$ partial isolation (\ref{def: Ay partial iso}) in the cancer and quality of life example.  Provided that these treatments are defined such that they can plausibly be developed and assigned in a future trial $G$, the conditional separable effect evaluated at $A_D=a_D$ is identified by design in that trial and trivially estimated by the mean difference in outcomes in the arm assigned $A_Y=1$ and $A_D=a_D$ versus $A_Y=0$ and $A_D=a_D$, among all those in these two arms not experiencing the post-treatment event.

\section{The Modified Treatment Assumption} 
\label{sec: modif treat}
We now give conditions under which the conditional separable effects provide an explanation of the mechanism by which the original treatment $A$ affects $Y$. Consider two studies: the current study where $A$ is randomly assigned and a future study where $A_Y$ and $A_D$ are jointly assigned. Following Stensrud et al \cite[Appendix A]{stensrud2019generalized}, for two variables $M_Y$ and $M_D$, suppose that the following conditions hold in these two studies:
\begin{align}
& \text{All effects of } A, A_Y \text{ and } A_D \text{ on }  Y \text{ and } D_k, k \in \{0,\dots, K\}, \text{ are intersected }  \nonumber \\
& \text{by } M_Y \text{ and } M_D, \text{ respectively, and} \nonumber \\
& \quad M^{a_Y=a,a_D}_Y = M^{a}_Y \quad \text{ for } a_D \in \{0,1\}, \nonumber \\ 
& \quad M^{a_Y,a_D=a}_D = M^{a}_D \quad  \text{ for } a_Y \in \{0,1\}. \label{ass: modif treat}
\end{align}
We refer to \eqref{ass: modif treat} as the \textit{modified treatment assumption}. It follows from \eqref{ass: modif treat} that jointly assigning $A_Y$ and $A_D$ to the same value $a$ leads to exactly the same values of $Y$ and $D_{k+1}, k \in \{0,\dots, K\}$ as assigning $A$ to $a$. Robins and Richardson \cite{robins2010alternative} introduced a decomposition assumption that is covered by our modified treatment assumption \eqref{ass: modif treat}: let the treatments $A_Y$ and $A_D$ constitute a decomposition of $A$ such that $A$ exerts all its effects on $Y$ and $\overline{D}_{K+1}$ through $A_Y$ and $A_D$, and the following determinism holds in the current study,
\begin{align}
    A\equiv A_D\equiv A_Y.
    \label{assumption: Determinsm}
\end{align} 
Their decomposition assumption is a special case of \eqref{ass: modif treat} by defining $A_Y \equiv M_Y$ and  $A_D \equiv M_D$. We can consider an example where \eqref{ass: modif treat} holds but not their decomposition assumption. Suppose that an old chemotherapeutic treatment ($A=1$) has an unpleasant side-effect that reduces quality of life compared to no treatment ($A=0$): it causes nausea by binding to certain (neurokinin 1) receptors in the brain. However, a modified version of this treatment ($A_Y=0,A_D=1$) can be created, where the chemical structure of the old treatment is slightly changed such that it no longer binds to the receptors in the brain ($M^{a_Y=0,a_D=1}_Y = M^{a=0}_Y$), but still exerts the cytotoxic effects on cancer cells and thus reduces mortality ($M^{a_Y=0,a_D=1}_D = M^{a=1}_D$). Suppose that the new drug satisfies \eqref{ass: modif treat}. Yet, it does not necessarily satisfy the treatment decomposition assumption \eqref{assumption: Determinsm} because it is not a (physical) decomposition of the old treatment. 

Under \eqref{ass: modif treat}, the conditional separable effect evaluated at $A_D=a_D$ \eqref{eq: conditional sep eff 1} is defined in the subset of the population with $D_{k+1}^{a_Y,a_D}=D_{k+1}^{a=a_D}=0$, which is simply identified in the current trial of Section \ref{sec: observed data structure} by those with $D_{K+1}=0$ and treatment assignment $A=a_D, a_D \in \{0,1\}$. 

Like $A_Y$ partial isolation, assumption \eqref{ass: modif treat} for a choice of $A_Y$ and $A_D$ must be justified by subject matter knowledge, and can, in principle, be falsified in a plausible six-arm randomized experiment (denoted $G'$) in which individuals are randomly assigned to $A$ (without assignment to $A_Y$ or $A_D$) or joint assignment to $A_Y$ and $A_D$ (without assignment to $A$). We can for example observe in this six-arm trial whether $E(V \mid A_Y=a, A_D=a)= E(V \mid A=a), \text{ for } a =\{0,1\} \text{ and } V \in \{Y,\overline{D}_{K+1}\}$, with failure of this equality falsifying \eqref{ass: modif treat}.

Causal graphs can display a necessary condition for the modified treatment assumption \eqref{ass: modif treat}; that is, whether $A$, $A_Y$ and $A_D$  exert all their effects on $Y$ and $D_{k+1}, k \in \{0,\dots, K\}$ through $M_Y$ and $M_D$. This is illustrated in Figure \ref{fig: decomposition 1}d, which describes a six-arm trial $G'$ that expands the 2 arm-trial in Figure \ref{fig: decomposition 1}a. Furthermore, the graph in Figure \ref{fig: decomposition 1}b can be interpreted as a transformation of the graph in Figure \ref{fig: decomposition 1}d that removes the node $A$, representing the data generating mechanism under $G'$ had we removed the two arms assigning $A$ (i.e.\ a four-arm trial $G$). It is not necessary to include $M_Y$ and $M_D$ in this reduced graph, because they are not common causes of any variable.

If we impose the treatment decomposition assumption of Robins and Richardson \cite{robins2010alternative}, we can represent the causal structure in an \textsl{extended causal DAG} \cite{robins2010alternative}; that is, a transformation of the original causal DAG representing $A$ and the components $A_Y$ and $A_D$, where the mechanisms by which these components individually operate on outcomes are encoded \cite{robins2010alternative,stensrud2019generalized}.  For example, Figure \ref{fig: decomposition 1}c is an extension of Figure \ref{fig: decomposition 1}a representing  \eqref{assumption: Determinsm}, with bold arrows representing deterministic relations. Under the decomposition assumption, Figure \ref{fig: decomposition 1}b can be interpreted as a $G$-transformation of the extended DAG in Figure \ref{fig: decomposition 1}c, where (i) the node $A$ (and any of its causes) are removed and (ii) all nodes are indexed by $G$, with $G$ again indexing the 4-arm trial discussed above.

\section{$A_Y$ partial isolation and meaningful effects of $A$ on $Y$}
\label{sec: Ay partial meaningful eff}
By explicitly considering modified treatments such that $A_Y$ partial isolation \eqref{def: Ay partial iso} and the modified treatment assumption holds, the investigator is forced to articulate what she means by an "effect of $A$ on $Y$ not through $\overline{D}_{K+1}$" under assumptions that are falsifiable in a future experiment. This thought process requires the investigator to be explicit about her notion of a causal mechanism, and allows the consideration of a well-defined causal effect: in particular, under these assumptions, the conditional separable effects quantify mechanisms by which $A$ affects $Y$ that can be entirely \emph{separated} from mechanisms by which $ A$ affects $D_{k+1}, k \in \{0,\dots, K\}$. We discuss this further in Section \ref{sec: full iso section}.  
 


If the investigator is unable to express a convincing story about modified treatments satisfying \eqref{def: Ay partial iso} and the modified treatment assumption, then the relevance of an effect of $A$ on $Y$ outside of its effect on the post-treatment event $D_{k+1}$ is ambiguous: the investigator has failed to give a plausible scientific argument as to how effects of $A$ on $Y$ can be disentangled from effects of $A$ on $D_{k+1}$. Yet, even if the investigator cannot define plausible modified treatments satisfying \eqref{def: Ay partial iso} and the modified treatment assumption at this moment in time, these assumptions may be justified in the future: modified treatments satisfying these assumptions might be revealed when more subject-matter knowledge becomes available. However, until a plausible story can be articulated such that \eqref{def: Ay partial iso} and the modified treatment assumption are satisfied, the practical relevance of considering any effect of $A$ on $Y$ outside of its effect on $D_{k+1}$ -- including the conventional principal stratum effect -- is unclear.  In this case, the investigators must accept that they do not understand how $A$ exerts such effects and have no way to assess whether such effects are operating in the data without reliance on assumptions that are impossible to ever challenge in real-life experiments.  In turn, without any ideas about such modified treatments, we cannot imagine interventions to avoid or leverage effects of $A$ on $Y$ outside of its effect on $D_{k+1}$.


\subsection{Example: Cancer Treatment and Quality of Life}
\label{sec: prostate revisited}
Returning to our cancer treatment and quality of life example, suppose that Figure \ref{fig: isolation conditions}a represents data generating assumptions on a trial that assigns treatment at baseline ($A=1$ is new chemotherapy, $A=0$ is standard chemotherapy) with $Y$ a quality of life measure at end of follow-up and $D_{k+1}$ an indicator of death by time $k+1$.  Suppose that there exist two modified treatments $A_Y$ and $A_D$ satisfying the following assumptions, which are also illustrated in Figure \ref{fig: isolation conditions}c: $A_D$ exerts effects on mortality $D_{k+1}, k = 0, \dots, K$, e.g.\ by destroying or reducing the growth of cancer cells and thereby preventing cancer progression ($Z_j,  j = 0, \dots, k+1$). By preventing cancer progression, the $A_D$ component may also reduce other health problems, e.g.\ due to metastases, which affect quality of life $Y$ (the path $A_D \rightarrow Z_1 \rightarrow Y$ in Figure  \ref{fig: isolation conditions}c). The other component $A_Y$ does not exert effects on mortality, because it has little to no activity against the cancer but may have side effects that adversely affect quality of life; for example it may interfere with the replication of epithelial mucosal cells, resulting in diarrhea and oral ulcers. Alternatively, this component may possibly have beneficial effects (say due to a decrease in diarrhea and oral ulcers). 

In this setting, the conditional separable effect evaluated at $a_D=1$ quantifies the treatment effect on quality of life outside of its effect on disease progression.  Specifically, this quantifies the effect of assignment to a current chemotherapy (e.g.\ $a =a_Y=a_D= 1$) versus a modified (hypothetical) therapy that contains the component of the current therapy that reduces mortality and disease progression ($a_D=1$), but does not contain the component that exerts effects on quality of life outside of mortality ($a_Y=0$). An improvement of quality of life under the modified therapy suggests that the new chemotherapy ($a=1$) contains a component that would be desirable to eliminate. 

$A_Y$ partial isolation would fail to hold in our example if $A_Y$ exerts effects on a common cause of $Y$ and $D_{k+1}, k = 0, \dots, K+1$, as illustrated in Figure \ref{fig: isolation conditions}d by the path $A_Y \rightarrow Z_1 \rightarrow D_2$.  One possible common cause could be quality of life at earlier times $k < K$. If the $A_Y$ component exerts effects on quality of life (which we now denote $Z_k$ for $k < K$) only after a minimal latent period that extends beyond the study period \cite{robins2008causal}, then $A_Y$ partial isolation may still be justified (e.g. the arrow from $A_Y$ into $Z_1$ in Figure \ref{fig: isolation conditions}d may be removed). Alternatively, if quality of life only exerts effects on mortality after a minimal latent period that extends beyond the study period then $A_Y$ partial isolation may be justified (e.g. the arrow from $Z_1$ into $D_2$ can be removed). 


\subsection{When the conditional separable effect is the direct effect of $A$}
\label{sec: full iso section}
While the conditional separable effects are well-defined under $A_Y$ partial isolation, this condition allows additional causal paths from $A$ to $Y$ that are not intersected by $D_{k+1}, k = 0, \dots K$, as illustrated in Figure \ref{fig: isolation conditions}c by the path $ A_D \rightarrow Z_1 \rightarrow Y$. However, suppose that, in a causal DAG representing the assumed data generating mechanism in the four-arm trial $G$,
\begin{align}
    & \text{the only causal paths from } A_D(G) \text{ to } Y(G),  \text{ are directed paths} \label{def: Ad partial iso}  \\ 
    &\text{intersected by } D_{k}(G), k \in \{1,\dots,K\}, \nonumber 
\end{align}
which we refer to as $A_D$ partial isolation \cite{stensrud2019generalized}. When both $A_D$ partial isolation \eqref{def: Ad partial iso} and $A_Y$ partial isolation \eqref{def: Ay partial iso} simultaneously hold, we say there is \textsl{full isolation} \cite{stensrud2019generalized}.  Under full isolation, the conditional separable effects capture \emph{all} causal paths from $A$ to $Y$ not intersected by $D_{k+1}, k = 0, \dots K$. Full isolation is represented in Figures \ref{fig: decomposition 1}b and \ref{fig: isolation conditions}b.

Returning to our running cancer treatment and quality of life example represented by Figure \ref{fig: isolation conditions}c, full isolation would hold under the stronger assumption that cancer progression $Z_k, k = 0, \dots, K+1$ does not affect quality of life $Y$, allowing removal of the arrow from $Z_1$ to $Y$. This assumption seems to be implausible for many cancer treatments: by preventing cancer progression, the treatment will not only reduce mortality, but also reduce other effects of progression, such as pain related to tumor growth. 


\section{Identifiability conditions}
\label{sec: identifiability conditions}
If we had data from a four-arm trial in which $A_Y$ and $A_D$ were randomly assigned and censoring is absent, conditions required to identify $\mathbb{E}( Y^{a_Y,a_D} \mid 
D^{a_Y,a_D}_{K+1}=0)$ for $a_Y,a_D \in \{0,1\}$ hold by design.  Consequently, these conditions also identify the conditional separable effects under the assumption of $A_Y$ partial isolation by
\begin{equation}
\mathbb{E} ( Y^{a_Y,a_D}\mid
D^{a_D}_{K+1}=0 )=  \mathbb{E} ( Y^{a_Y,a_D}\mid
D^{a_Y,a_D}_{K+1}=0 ). \label{ass: parameter equality}
\end{equation} 
However, in the two-arm trial in which only the original treatment $A$ is randomly assigned, we are not guaranteed identification of $\mathbb{E}( Y^{a_Y,a_D} \mid D^{a_Y,a_D}_{K+1}=0)$ when $a_Y \neq a_D$ in this trial even when censoring is absent.

We now consider a set of conditions, beyond the conventional exchangeability, consistency and positivity conditions that hold by design in the two-arm trial of Section \ref{sec: observed data structure} (reviewed in Appendix \ref{app sec: proof of id}, Theorem \ref{theorem g formula}). Under $A_Y$ partial isolation and the modified treatment assumption, these conditions are sufficient to identify $\mathbb{E} ( Y^{a_Y,a_D}\mid
D^{a_Y,a_D}_{K+1}=0 )$ when $a_Y \neq a_D$ using only data from this existing trial.  By \eqref{ass: parameter equality}, this allows identification of the conditional separable effects.
\begin{enumerate}
\item[1.] Positivity: 
\begin{align}
& f_{\overline{L}_k,D_{k+1}}(\overline{l}_k,0) > 0  \implies \nonumber\\ 
& \quad  \Pr(A=a|D_{k+1}=0,\overline{L}_k=\overline{l}_k)>0, k \in \{0,\ldots,K\}, a\in\{0,1\}. \label{eq: positivity of A part 2}
\end{align}%
Assumption \eqref{eq: positivity of A part 2} states that for any possibly observed level of the time-varying covariate history among those surviving through each follow-up time, there exist individuals with $A=1$ and individuals with $A=0$. Assumption \eqref{eq: positivity of A part 2} does not hold by design in a randomized experiment, but it can be assessed in the observed data. 

\item[2.] Dismissible component conditions:
\begin{align}
& Y(G) \independent A_D(G) \mid A_Y(G), D_{K+1}(G)=0, \overline{L}_K(G), \label{ass: delta 1}\\
 & D_{k+1}(G) \independent A_Y(G) \mid A_D(G), D_{k}(G)=0, \overline{L}_k(G), \label{ass: delta 2} \\
& L_{k+1}(G) \independent A_Y(G) \mid A_D(G), D_{k+1}(G)=0, \overline{L}_{k}(G),  \label{ass: delta 3}
\end{align} 
for all $k \in \{0,\dots, K\}$ where $Y(G)$, $ D_{k+1}(G)$ and $L_{k+1}(G)$ are values of the outcome of interest, the conditioning event and the measured covariates at $k+1$ had we implemented the four-arm trial $G$ that jointly assigns combinations of $A_Y$ and $A_D$. If the dismissible component conditions \eqref{ass: delta 1}-\eqref{ass: delta 3}, then $A_Y$ partial isolation holds (See Lemma \ref{lemma: diss and AY partial} in Appendix \ref{sec: dismissible comp imply Ay partial}). The dismissible component conditions \eqref{ass: delta 1}-\eqref{ass: delta 3} can be assessed in $G$-transformation graphs as discussed in Section \ref{sec: modif treat}; i.e.\ causal DAGs that represent the assumed data generating mechanism in the future four-arm trial $G$. For example Figure \ref{fig: dismissible component conditions} depicts the assumption that there exist measured (e.g.\ $L_1$) and unmeasured (e.g.\ $U_{L,Y}$ or $U_{L,D}$) common causes of $Y$ and $\overline{D}_2$ in $G$. The dismissible component conditions hold in Figures \ref{fig: dismissible component conditions}a-c, but \eqref{ass: delta 1} is violated in Figures \ref{fig: dismissible component conditions}d-e and \eqref{ass: delta 3} is violated in Figure \ref{fig: dismissible component conditions}f. Also, $A_Y$ partial isolation is violated in Figures \ref{fig: dismissible component conditions}c,d,f.   
\end{enumerate}

Under these assumptions, we can identify the conditional counterfactual mean $\mathbb{E}  ( Y^{a_Y,a_D} \mid 
D^{a_Y,a_D}_ {K+1}=0)$ by the following function of the observed data,
\begin{align}
& \frac{\sum_{\overline{l}_K}  \mathbb{E}  ( Y \mid D_{K+1}=0,  \overline{L}_{K}=\overline{l}_{K}, A=a_Y)  f_{\underline{L}_1,\overline{D}_{K+1} \mid L_0, A}(\underline{l}_{1},0 \mid l_0,  a_D)f_{L_0}(l_0)}{P(D_{K+1}=0 \mid A=a_D)} . \nonumber \\ 
\label{eq: id formula no cens}
\end{align}
See Appendix \ref{app sec: proof of id} (Theorem \ref{theorem g formula}) for a proof that also covers settings where individuals can be lost to follow-up\footnote{In Appendix \ref{app sec: proof of id}, the proof for the identification formula of  $\mathbb{E}  ( Y^{a_Y,a_D} \mid 
D^{a_Y,a_D}_ {K+1}=0)$ is  given in a more general setting where $A_Y$ partial isolation is not required to hold. However, $A_Y$ partial isolation is a necessary condition for the conditional separable effects to be well-defined.}.  We say that \eqref{eq: id formula no cens} is the g-formula for $\mathbb{E}  ( Y^{a_Y,a_D} \mid 
D^{a_Y,a_D}_{K+1}=0)$ \cite{robins1986new}. 

Importantly,  $A_Y$ partial isolation clarifies when a contrast of conditional counterfactual outcome means, 
\begin{align}
    & \mathbb{E}  ( Y^{a_Y=1,a_D} \mid 
D^{a_Y=1,a_D}_{K+1}=0) \text{ vs. } \mathbb{E}  ( Y^{a_Y=0,a_D} \mid 
D^{a_Y=0,a_D}_{K+1}=0), \label{eq: conditional means agnostic}
\end{align}
can be interpreted as a conditional causal effect: even if treatment $A_Y$ is temporally ordered before $D_{K+1}$, $A_Y$ partial isolation allows us to topologically order the component $A_Y$ \textit{after} $D_{K+1}$. That is, our estimands are isomorphic to estimands defined by interventions where we first assign a treatment $A_D$ and then, after $D_{K+1}$ occurs, we subsequently assign $A_Y$. This can be seen on a SWIG with minimal labelling \cite{richardson2013single} consistent with $A_Y$ partial isolation, as in Figure \ref{fig: swig minimal labelling}a. However, in Appendix \ref{app sec: proof of id} we show that $A_Y$ partial isolation is not in general needed to identify $\mathbb{E}( Y^{a_Y,a_D} \mid D^{a_Y,a_D}_{K+1}=0)$; $A_Y$ partial isolation is only needed to ensure that \eqref{eq: conditional means agnostic} is an average over individual level causal effects.

\subsection{Related works on identification of path specific effects}
\label{sec: related works}
There is an intimate link between identification results for our estimands and path-specific effects, as recently discussed in Robins et al \cite{robins2020interventionist}: when there is no so-called \textit{recanting witness} \cite{avin2005identifiability, shpitser2013counterfactual} (see also Stensrud et al \cite[Section 6.5]{stensrud2019generalized}), the (cross-world) counterfactuals defining a path specific effect are equal to (single-world) counterfactuals defined by interventions on nodes in an extended causal graph, like Figure \ref{fig: decomposition 1}c. Thus, the identification formulas for effects defined by interventions in an extended causal graph are equal to identification formulas derived for certain path-specific effects \cite{robins2020interventionist}. Malinsky et al \cite{malinsky2019potential} recently derived a complete algorithm for identification of conditional path specific distributions, which is inspired by the general identification theory of Shpitser \cite{shpitser2013counterfactual}. However, while the identification formulas for these estimands can be equal, the effects being identified are different: they refer to interventions on different variables (see Robins et al \cite{robins2020interventionist} for a detailed discussion). Unlike our current results, the previous works did not consider conditions under which a contrast of conditional counterfactual outcome parameters can be interpreted as a conditional causal effect (i.e.\ a contrast of counterfactual outcome means in the same set of individuals).

\subsection{New identification results for principal stratum effects}
\label{sec: sep eff vs ps}

For the principal stratum effect \eqref{eq: SACE} to target the same subpopulation as the conditional separable effects \eqref{eq: conditional sep eff 1} under both $a_D=0$ and $a_D=1$, we must make the additional strong assumption that $D^{a=1}_{k+1} = D^{a=0}_{k+1}$, for $k \in \{0,\dots, K\}$. For example, under the assumption that lack of arrows in a causal graph means no individual level causal effect, this requires that the arrow from $A$ into $D_1$ in Figure \ref{fig: decomposition 1}a or the arrows from $A$ into $D_1$ and $D_2$ in Figure \ref{fig: isolation conditions}a should be removed. In our cancer and quality of life example, this assumption requires survival to be identical under the new ($a=1$) and standard ($a=0$) chemotherapy.  This underscores that the conditional separable effects and the principal stratum effect \eqref{eq: SACE} -- specifically the \textit{survivor average causal effect} (SACE) in this example -- are substantially different estimands; these estimands will only be equivalent in the restrictive setting where $A$ exerts no effect on mortality.  

Suppose instead we make the weaker assumption that the effect of $A$ on $D_{k+1}$ is \emph{monotone}: without loss of generality, this is the assumption that $D^{a=1}_{k+1} \leq D^{a=0}_{k+1}$ in all individuals for all $k$.  Under this assumption and full isolation, the principal stratum effect \eqref{eq: SACE} is equivalent to the conditional separable effect under $a_D = 0$,
\begin{align}
& \mathbb{E} ( Y^{a=1} - Y^{a=0} \mid 
D^{a=1}_{K+1}=D^{a=0}_{K+1}=0 ) \nonumber \\
= & \mathbb{E} ( Y^{a=1} - Y^{a=0} \mid 
D^{a=0}_{K+1}=0 ) \quad \text{ by monotonicity} \nonumber \\
= & \mathbb{E} ( Y^{a_Y=1,a_D=1} \mid 
D^{a_Y=0,a_D=0}_{K+1}=0 ) - \mathbb{E} ( Y^{a_Y=0,a_D=0} \mid
D^{a_Y=0,a_D=0}_{K+1}=0 )  \text{ by } \eqref{ass: modif treat}  \nonumber \\
= & \mathbb{E} ( Y^{a_Y=1,a_D=0} \mid 
D^{a_Y=0,a_D=0}_{K+1}=0 ) - \mathbb{E} ( Y^{a_Y=0,a_D=0} \mid
D^{a_Y=0,a_D=0}_{K+1}=0 )  \text{ by } \eqref{def: Ad partial iso} \nonumber \\ 
= &\mathbb{E} ( Y^{a_Y=1,a_D=0}  - Y^{a_Y=0,a_D=0} \mid
D^{a_Y=0,a_D=0}_{K+1}=0 )  \quad   \nonumber  \\
= &\mathbb{E} ( Y^{a_Y=1,a_D=0}  - Y^{a_Y=0,a_D=0} \mid
D^{a_D=0}_{K+1}=0 ) \text{ by }  \eqref{def: Ay partial iso}.  
\label{eq: contrast mono}
\end{align}

It follows directly from \eqref{eq: contrast mono} that full isolation and conditions \eqref{eq: positivity of A part 2}-\eqref{ass: delta 3} are sufficient to identify the principal stratum effect \eqref{eq: SACE} under monotonicity in the two-arm trial. Thus, our identification results supplement the few suggested identifiability assumptions that are sufficient for identification of principal stratum effects, such as the survivor average causal effect, in the presence of measured, possibly time-varying, common causes of the event of interest and the post-treatment conditioning event \cite{ tchetgen2014identification}. Furthermore, except for monotonicity, we only rely on assumptions that can be falsified (rejected) in a future randomized experiment. 


\section{Estimation}
\label{sec: estimation}
Let $\nu_{a_Y,a_D} $ denote the g-formula \eqref{eq: id formula no cens}. Here we consider various estimators for this parameter in the absence of censoring.  Extensions to allow censoring are given in Appendix \ref{sec: appendix censoring}.   

\subsection{Outcome regression estimator} A simple outcome regression estimator $\hat{\nu}_{or,a_Y,a_D}$ of $\nu_{a_Y,a_D} $ is the solution to the estimating equation $\sum_{i=1}^{n}U_{or,i}(\nu_{a_Y,a_D},\hat{\theta})=0$ with respect to $\nu_{a_Y,a_D}$, with
\begin{align*}
U_{or,i}(\nu_{a_Y,a_D},\hat{\theta}) =   I(A_i=a_D) (1-D_{k+1,i})  \left( \mathbb{E}(Y \mid D_{K+1}=0, A=a_Y, \overline{L}_{K} ; \hat{\theta}) - \nu_{a_Y,a_D} \right),
\end{align*}
where $\mathbb{E}(Y \mid D_{K+1}=0, A=a_Y, \overline{L}_{K} ; \theta)$ is a parametric model for $\mathbb{E}  ( Y \mid D_{K+1}=0,  \overline{L}_{K}, A=a_Y)$ indexed by the parameter ${\theta}$ and assume $ \hat{\theta}$ is its MLE.  The estimating equation $U_{or,i}(\nu_{a_Y,a_D},\hat{\theta})$ has mean zero and the estimator $\hat{\nu}_{or,a_Y,a_D}$ is consistent provided that this model is correctly specified. 

\subsection{Weighted estimator}

Alternatively, define the weighted estimator $\hat{\nu}_{ipw,a_Y,a_D} $ of $\nu_{a_Y,a_D} $ as the solution to the estimating equation $\sum_{i=1}^{n}U_{ipw,i}(\nu_{a_Y,a_D},\hat{\alpha})=0$ with respect to $\nu_{a_Y,a_D}$, with
\begin{align*}
& U_{ipw,i}(\nu_{a_Y,a_D},\hat{\alpha}) =   I(A_i=a_Y) (1-D_{k+1,i}) \hat{W}_{i}(a_Y,a_D;\hat{\alpha}) \left(  Y_{i}  - \nu_{a_Y,a_D} \right), \nonumber \\
\end{align*}
such that
\begin{align*}
 \hat{W} (a_Y,a_D;\hat{\alpha}) &= \frac{  f_{\underline{L}_1,\overline{D}_{K+1} \mid L_0, A}(\underline{L}_{1},0 \mid L_0, a_D; \hat{\alpha}) }{  f_{\underline{L}_1,\overline{D}_{K+1} \mid L_0, A}(\underline{L}_{1},0 \mid L_0, a_Y ; \hat{\alpha}) }
\end{align*}
is an estimator of 
\begin{align}
W (a_Y,a_D) = \frac{  f_{\underline{L}_1,\overline{D}_{K+1} \mid L_0, A}(\underline{L}_{1},0 \mid L_0, A = a_D) }{ f_{\underline{L}_1,\overline{D}_{K+1} \mid L_0, A}(\underline{L}_{1},0 \mid L_0,  A = a_Y) }, 
\label{eq: weight w_1}
\end{align}
where $ f_{\underline{L}_1,\overline{D}_{K+1} \mid L_0, A}(\underline{L}_{1},0 \mid L_0, A = a; \alpha)$ is a parametric model for $ f_{\underline{L}_1,\overline{D}_{K+1} \mid L_0, A}(\underline{L}_{1},0 \mid L_0, A = a)$ indexed by the parameter $ \alpha$ and assume $ \hat{\alpha}$ is its MLE. If this model is correctly specified, then $\hat{\nu}_{ipw,a_Y,a_D} $ is a consistent estimator for $\nu_{a_Y,a_D} $, which follows from an alternative representation of the g-formula that is derived in Appendix \ref{app sec: alternative id formula proof} (Theorem \ref{theorem: weighted id formula 1}). 

In practice, parameterizing the terms in $W(a_Y,a_D)$ requires care when $\overline{L}_K$ is high-dimensional. However, we can re-express  $\hat{W} ( \cdot )$ as a product of the terms
\begin{align*}
\hat{W}_{D} (a_Y,a_D;\hat{\alpha}_D) &= \frac{\prod_{j=0}^{K}  \Pr(D_{j+1}=0 \mid   \overline{L}_{j} ,D_{j}=0 ,  A = a_D; \hat{\alpha}_D) }{ \prod_{j=0}^{K} \Pr(D_{j+1}=0 \mid   \overline{L}_{j} ,D_{j}=0 ,  A = a_Y ; \hat{\alpha}_D) }, \\
\hat{W}_{L} (a_Y,a_D;\hat{\alpha}_{L}) &=
\frac{\prod_{j=0}^{K-1} f_{\overline{L}_{j+1} \mid \overline{D}_{j+1},\overline{L}_{j}, A}(L_{j+1}  \mid  0, \overline{L}_{j},  a_D; \hat{\alpha}_{L})}{\prod_{j=0}^{K-1} f_{\overline{L}_{j+1} \mid \overline{D}_{j+1},\overline{L}_{j}, A}(L_{j+1}  \mid  0, \overline{L}_{j}, a_Y; \hat{\alpha}_{L})}, 
\end{align*}
where $\Pr(D_{j+1}=0 \mid   D_j=0,\overline{L}_{j},  A = a; \alpha_D)$ is a pooled over time model for $\Pr(D_{j+1}=0 \mid   D_j=0, \overline{L}_{j},  A = a)$ indexed by the parameter $ \alpha_D$, and $\hat{\alpha}_D$ its MLE. The parameter $\alpha_{L}$ analogously indexes a pooled over time model for the conditional density $f_{\overline{L}_{j+1} \mid \overline{D}_{j+1},\overline{L}_{j}, A}(L_{j+1}  \mid  0, \overline{L}_{j}, a)$, with $\hat{\alpha}_{L}$ its MLE. Off the shelf software can be used to estimate  $\Pr(D_{j+1}=0 \mid   D_j=0,\overline{L}_{j},  A = a)$.  More complex computation (and strong parametric assumptions) may be required to consistently estimate $f_{\overline{L}_{j+1} \mid \overline{D}_{j+1},\overline{L}_{j}, A}(L_{j+1}  \mid  0, \overline{L}_{j}, a)$ when $L_j$ is high-dimensional.   In this case, we can alternatively fit a model directly for the likelihood ratio 
$$\frac{f_{\overline{L}_{j+1} \mid \overline{D}_{j+1},\overline{L}_{j}, A}(L_{j+1}  \mid  0, \overline{L}_{j},  a_D)}{f_{\overline{L}_{j+1} \mid \overline{D}_{j+1},\overline{L}_{j}, A}(L_{j+1}  \mid  0, \overline{L}_{j}, a_Y)},$$
for $j \in \{0, \dots, K\}$, e.g. a proportional likelihood ratio model as a more parsimonious function of $\overline{L}_j$ and possibly $j$ which avoids distributional assumptions on $\overline{L}_{j+1}$ \cite{luo2012proportional}.


\subsection{Doubly robust estimator}
\label{sec: doubly robust}
We can alternatively consider a doubly robust estimator derived from the nonparametric influence function in Appendix \ref{app sec: influence doubly} \cite{robins1994estimation, van2003unified, van2000asymptotic}. This estimator $\hat{\nu}_{dr,a_Y,a_D}$ is the solution to the estimating equation

\begin{align*}
U_{dr}(\nu_{a_Y,a_D},\hat{\alpha},\hat{\theta}) & = \Bigg\{ \frac{I( A=a_D) }{P( A=a_D)} (1-D_{K+1}) \mathbb{E}(Y \mid D_{K+1}=0, A=a_Y, \overline{L}_{K};\hat{\theta})  \\  
   & \  +  \frac{I( A=a_Y)}{P( A=a_Y)} \frac{f_{\underline{L}_1,\overline{D}_{K+1} \mid L_0, A}(\underline{L}_{1},0 \mid L_0, a_D);\hat{\alpha})}{f_{\underline{L}_1,\overline{D}_{K+1} \mid L_0, A}(\underline{L}_{1},0 \mid L_0, a_Y;\hat{\alpha})} (1- D_{K+1}) \\
   &  \ \times [ Y - \mathbb{E}( Y \mid D_{K+1}=0,A=a_Y, \overline{L}_{K})   ]   \Bigg\} \frac{1}{\hat{\mathbb{E}}(  1-D_{K+1} \mid    A=a_D )} \\
   & -  \nu_{a_Y,a_D}, %
\end{align*}
with $\hat{\mathbb{E}}(  1-D_{K+1} \mid    A=a_D)$ simply a sample mean. In Appendix \ref{app sec: influence doubly} we show that $\hat{\nu}_{dr,a_Y,a_D}$  is doubly robust; that is, it is consistent if either the models for $f_{\underline{L}_1,\overline{D}_{K+1} \mid L_0, A}(\underline{L}_{1},0 \mid L_0, a)$ or $\mathbb{E}(Y \mid D_{K+1}=0, \overline{L}_{K},  A = a)$ are correctly specified, but not necessarily both. Analogous to the weighted estimator $\hat{\nu}_{IPW,a_Y,a_D} $, consistency of $\hat{\nu}_{dr,a_Y,a_D}$ may also be acheived by correctly specifying a model for the likelihood ratio in place of $f_{\underline{L}_1,\overline{D}_{K+1} \mid L_0, A}(\underline{L}_{1},0 \mid L_0, a)$, as described in the previous section.  Note that the estimating equation for $\hat{\nu}_{dr,a_Y,a_D}$ can serve as the basis for constructing estimators that use machine learning to estimate nuisance parameters in place of parametric models. 



\section{Data example: the Southwest Oncology Group Trial}
\label{sec: data example}

As an illustration, we analyzed data from the Southwest Oncology Group (SWOG) Trial \cite{petrylak2004docetaxel} that randomly assigned men with refractory prostate cancer to either of two chemotherapies, Docetaxel and Estramustine (DE) or Mitoxantrone and Prednisone (MP). Our dataset, which included 487 patients aged 47 to 88 years, has previously been used to compare outcomes under DE versus MP on health related quality of life in the principal stratum of always survivors \cite{wang2017identification, ding2011identifiability, yang2018using}, that is, to estimate the survivor average causal effect. Yet, the practical relevance of the survivor average causal effect is ambiguous, as we discussed in Section \ref{sec: Ay partial meaningful eff}.

The conditional separable effects quantify notions of causal mechanisms and can be used to conceive improved treatments. Thus, our target of inference was a conditional separable effect of DE versus MP on the outcome $Y$, defined as the change in quality of life between baseline and one year of follow-up, similar to the previous reports evaluating the SACE \cite{wang2017identification, ding2011identifiability, yang2018using}. We considered a modification of the original treatments DE ($A=1$) and MP ($A=0$), where $A_D=1$ indicates receiving only the Docetaxel component of DE and $A_D=0$ indicates receiving only Mitoxantrone component of MP. Both Docetaxel and Mitoxantrone are chemotherapeutic agents that can reduce proliferation of cancer cells, and, potentially, progression of the disease ($Z_k$). Further, let $A_Y=1$ indicate receiving only Estramustine component of DE ($A_Y=1$) and $A_Y=0$ indicate receiving only the Prednisone component of MP. Assume that $A_Y$ partial isolation is satisfied. Clearly, this condition is not guaranteed by design, but may be plausible for these choices of $A_Y$ and $A_D$.  Specifically, neither Estramustine, which consists of estrogen and chemotherapeutic medication, nor Prednisone, a steroid, are given to reduce mortality. However, these components can potentially give pain relief (palliation) \cite{petrylak2005future}, but may also have side-effects, including nausea, fatigue and vomiting \cite{autio2012therapeutic, petrylak2004docetaxel, petrylak2005future}, which in turn can affect quality of life. Therefore, it is not clear which of these components are most effective at improving quality of life.

The causal structure of this example is represented in Figure \ref{fig: isolation conditions}c, satisfying $A_Y$ partial isolation, where the  causal path $A_D \rightarrow Z_1 \rightarrow Y$ illustrates that $A_D$, the chemotherapeutic components Docetaxel or Mitoxantrone, can affect quality of life but only through effects on the cancer progression indicator $Z_k$.





After $K+1=12$ months of follow-up, 0.79 (95\% CI:  $[0.73, 0.84]$) survived in the DE arm compared to 0.72 (95\% CI:  $[0.66, 0.78]$) in the MP arm, resulting in an estimated additive causal effect of 0.07 (95\% CI:  $[-0.01, 0.15]$) of treatment $A$ on 12-month survival ($D_{12}$). Of the 368 survivors at 12 months, 152 had a measure of the outcome quality of life measured (i.e.\ were uncensored). All analyses were adjusted for the measured baseline covariates age, race, anticipated of prognosis, bone pain and performance status ($L_0$) as well as an indicator of disease progression by $k$ ($L_k$). These point estimates are in line with a possible improvement in mortality under DE compared to MP.  By our hypotheses above, this suggests that the effect of a modified version of DE that replaces the Estramustine component with Prednisone (i.e. the joint treatment $A_D = 1,A_Y=0$) is of interest with respect to quality of life.

The estimated mean of $Y$ among survivors was -4.4 (95\% CI:  $[-8.6, -0.5]$) units in the DE arm ($a_D=a_Y=1$) and  -9.1 (95\% CI:  $[-14.0, -3.9]$) units in the ME arm ($a_D=a_Y=0$) (Table \ref{tab: effect estimates}). However, we cannot assign a causal interpretation to a contrast of these estimates, as discussed in Section \ref{sec: intro}. Therefore we estimated the conditional separable effect for $a_D = 1$ (i.e. the comparison of outcomes under Estramustine versus Prednisone both given with Docetaxel) using the regression estimator $\hat{\nu}_{or,a_Y,a_D} $, the  weighted estimator $\hat{\nu}_{ipw,a_Y,a_D} $ based on the weights $ \hat{W} (a_Y,a_D;\hat{\alpha})$ as well as the doubly-robust estimator $\hat{\nu}_{dr,a_Y,a_D}$.  We assumed the following parametric models for the nuisance parameter $f_{\underline{L}_1,\overline{D}_{K+1} \mid L_0, A}(\underline{L}_{1},0 \mid L_0, a_D)$, 
\begin{align*}
& \text{logit} [\Pr(D_{12}=1 \mid A=0, L_0,L_{11}; \alpha_{D_1})]  = \alpha_{D,0} + \alpha'_{D,1}L_0 +\alpha_{D,2}L_{11}, \nonumber \\
& \text{logit} [\Pr(D_{12}=1 \mid A=1, L_0,L_{11}; \alpha_{D_2})]  = \alpha_{D,3} + \alpha'_{D,4}L_0 +\alpha_{D,5}L_{11}, \\ 
& \text{logit} [\Pr(L_{11}=1 \mid A, L_0; \alpha_{L})] = \alpha_{L,0} + \alpha_{L,1}A +\alpha'_{L,2}L_0,
\end{align*}
and the following model for $\mathbb{E}(Y \mid D_{12}=0, A=a_Y, \overline{L}_{11}=\overline{l}_{11})$,
\begin{align*}
& \mathbb{E}(Y \mid D_{K+1}=0, A=a, \overline{L}_{K}=\overline{l}_{K}; \theta) = \theta_{0} + \theta_{1}A +\theta'_{2}L_0 + +\theta_{3}L_{11}.
\end{align*}

Under the hypothetical treatment that contains the component of DE that affects mortality (Docetaxel, $a_D=1$), but contains the component of MP that potentially affects quality of life outside of mortality (Prednisone, $a_Y=0$), the estimated mean outcome after 12 months was $-6.1$ ($95\%$ CI $[-11.8, -0.7]$) using the simple regression estimator $\hat{\nu}_{or,a_Y,a_D}$ (Table \ref{tab: effect estimates}). The estimated additive conditional separable effect for $a_D=1$ was equal to $-1.7$ ($95\%$ CI $[-5.1, 8.1]$).

Similarly, using the simple weighted estimator $\hat{\nu}_{ipw,a_Y,a_D}$ the estimated mean outcome after 12 months was $-5.5$ ($95\%$ CI $[-9.7, -0.4]$) (Table \ref{tab: effect estimates}), and the estimated additive conditional separable effect for $a_D=1$ was equal to $-1.1$ ($95\%$ CI $[-5.6, 6.6]$). We also implemented the doubly robust estimator $\hat{\nu}_{dr,a_Y,a_D}$, which gave an estimated outcome mean of $-5.6$ ($95\%$ CI $[-11.5, -0.2]$) (Table \ref{tab: effect estimates}), and an estimated conditional separable effect for $a_D=1$ equal to $-1.7$ ($95\%$ CI $[-5.1, 8.4]$). Thus, the estimates from $\hat{\nu}_{or,a_Y,a_D}$, $\hat{\nu}_{ipw,a_Y,a_D}$ and $\hat{\nu}_{dr,a_Y,a_D}$ are similar. The confidence intervals are wide, and there is no clear evidence that the modified drug with Docetaxel and Prednisone would lead to improved quality of life compared to DE in this study. These conclusions, however, rely on the assumption of $A_Y$ partial isolation and the identifiability conditions in Section \ref{sec: identifiability conditions}. In particular, the presence of unmeasured common causes of mortality and quality of life would violate these assumptions. 

As discussed in Section \ref{sec: sep eff vs ps}, under the additional assumption of monotonicity, our results also allow identification and estimation of the SACE. However, we have already argued that the practical relevance of the SACE is ambiguous. Under the additional assumption of a monotonic effect of $a_D$ on mortality, i.e.\ $D^{a_D=1}_{12} \leq D^{a_D=0}_{12}$, the SACE is identified by the same functional as the conditional separable effect where $a_D$ is fixed to 0,
\begin{align*}
    \mathbb{E}(Y^{a_Y=0,a_D=0} \mid D^{a_Y=0,a_D=0}_{12}=0) - \mathbb{E}(Y^{a_Y=1,a_D=0} \mid D^{a_Y=1,a_D=0}_{12}=0),
\end{align*}
that is, the conditional separable effect among those who would survive under Mitoxantrone, regardless of whether they received Prenisone or Estramustine. In contrast, we have argued that the separable effect among those who would survive under Docetaxel (where $a_D$ is fixed to 1) is of primary interest. Furthermore, we do not believe that monotonicity holds in our example, i.e.\ that Mitoxantrone does not improve survival in any single individual: different individuals can experience different effects of Docetaxel and Mitoxantrone on survival, e.g.\ due to differences in the cancer cells (subtypes) across individuals.  

The only form of censoring in this data set was due to missing values of $Y$.  In Appendix \ref{sec: appendix censoring},  we describe how censoring was adjusted for in the analysis. 

\begin{table}[ht]
\def~{\hphantom{0}}
\begin{tabular}{llrl}
 \\
 \hline
\textbf{Estimand} & \textbf{Estimator} & \textbf{Estimate} & \textbf{(95\% CI)} \\
 \hline 
$\mathbb{E}(Y^{a=1} \mid D^{a=1}_{12}=0)$ & Non-parametric & -4.4 & (-8.6, -0.5) \\ 
  $\mathbb{E}(Y^{a=0} \mid D^{a=0}_{12}=0)$ & Non-parametric & -9.1 & (-14.0, -3.9) \\ 
    $\mathbb{E}(Y^{a_Y=0,a_D=1} \mid D^{a_Y=0,a_D=1}_{12}=0)$ & $\hat{\nu}_{or,a_Y,a_D} $ &  -6.1 & (-11.8, -0.7)  \\ 
  $\mathbb{E}(Y^{a_Y=0,a_D=1} \mid D^{a_Y=0,a_D=1}_{12}=0)$ & $\hat{\nu}_{ipw,a_Y,a_D} $  & -5.5 & (-9.7, -0.4) \\ 
   $\mathbb{E}(Y^{a_Y=0,a_D=1} \mid D^{a_Y=0,a_D=1}_{12}=0)$ & $\hat{\nu}_{dr,,a_Y,a_D} $  & -5.6 & (-11.5, -0.2) \\ 
  \\
\end{tabular}
\caption{Estimates of changes in health related quality of life 12 months after baseline. The percentile based confidence intervals are derived from 500 bootstrap samples.}
\label{tab: effect estimates}
\end{table}

\section{Discussion}
\label{sec: discussion}
This work presents new estimands, conditional separable effects, in settings where the outcome of interest is only defined conditional on a post-treatment event status. To define the conditional separable effects, we relied on the assumption of $A_Y$ partial isolation and the existence of two modified treatments that exert effects through distinct causal pathways \cite{robins2010alternative, stensrud2019separable,stensrud2019generalized, didelez2018defining}. 

Principal stratum effects have often been used to quantify causal effects on an outcome that is only of interest conditional on a post-treatment variable, such as survival \cite{robins1986new, frangakis2002principal} or having an infection \cite{hudgens2003analysis}, but the practical relevance of these estimands has been questioned, because (i) they are defined in an unknown subpopulation that may not even exist \cite{robins2007principal,joffe2011principal, young2018choice}, and (ii) estimating the principal stratum effects requires unfalsifiable assumptions that cannot be verified even in any randomized experiment \cite{robins2010alternative}. 

Unlike principal stratum effects, the conditional separable effects can be identified without relying on unfalsifiable independence assumptions or hypothetical interventions to prevent death; the conditional separable effects can be identified in a Finest Fully Randomized Causally Interpreted Structured Tree Graph (FFRCISTG) model \cite{robins1986new, robins2010alternative, richardson2013single}.  As such, the identifiability conditions can be scrutinized in a future experiment in which the treatment components $A_Y$ and $A_D$ are randomly assigned. Unlike principal stratum effects, the conditional separable effects are not restricted to an unknown subpopulation of unknown size, but the two versions of the conditional separable effects (evaluated under either $a_D=0$ or $a_D=1$) are defined in the (observable) subsets of individuals who would have a particular value of a post-treatment variable under assignment $A=0$ and $A=1$, respectively. 

Thus, the conditional separable effects overcome key concerns with principal stratum estimands, such as the survivor average causal effects. Yet the definition of the conditional separable effects hinges on $A_Y$ partial isolation, which must be justified in each practical setting. However, this is not a limitation of the estimand itself: unless we conceive plausible modified treatments meeting $A_Y$ partial isolation, it is not clear how to disentangle treatment effects on the post-treatment variable (e.g.\ death) and treatment effects on the outcome of interest. Furthermore, our results help researchers clarify their thinking about future, potentially improved treatments. The study of these future treatments is of scientific and public health interest, even if $A_Y$ partial isolation is falsified in a future experiment evaluating these treatments.





\clearpage
\bibliography{references}
\bibliographystyle{unsrt}

\clearpage

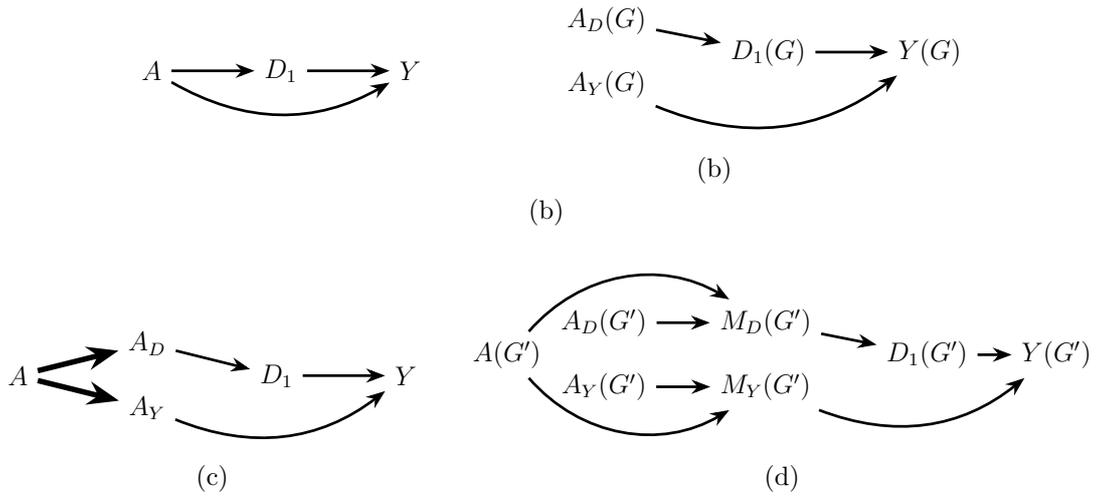
\begin{figure}
    \centering
\subfloat[]{
\scalebox{0.85}{
\begin{tikzpicture}
\begin{scope}[every node/.style={thick,draw=none}]
    \node (A) at (1,0) {$A$};
    \node (D1) at (3,0) {$D_1$};
    \node (Y) at (5,0) {$Y$};
     \node (NULL) at (1,-1) {$ $};
\end{scope}

\begin{scope}[>={Stealth[black]},
              every node/.style={fill=white,circle},
              every edge/.style={draw=black,very thick}]
	\path [->] (A) edge (D1);
	\path [->] (A) edge[bend right] (Y);
    \path [->] (D1) edge (Y);
\end{scope}
\end{tikzpicture}
} 
\subfloat[]{
\scalebox{0.85}{
\begin{tikzpicture}
\begin{scope}[every node/.style={thick,draw=none}]
    \node (A) at (-1,0) {$ $};
    \node (Ay) at (1,-0.5) {$A_Y(G)$};
	\node (Ad) at (1,0.5) {$A_D(G)$};
    \node (D1) at (3.5,0) {$D_1(G)$};
    \node (Y) at (6,0) {$Y(G)$};
\end{scope}

\begin{scope}[>={Stealth[black]},
              every node/.style={fill=white,circle},
              every edge/.style={draw=black,very thick}]
	\path [->] (Ad) edge (D1);
	\path [->] (Ay) edge[bend right] (Y);
    \path [->] (D1) edge (Y);
\end{scope}
\end{tikzpicture}
}}} 
\\ 
\subfloat[]{
\scalebox{0.85}{
\begin{tikzpicture}
\begin{scope}[every node/.style={thick,draw=none}]
    \node (A) at (-1,0) {$A$};
    \node (Ay) at (1,-0.5) {$A_Y$};
	\node (Ad) at (1,0.5) {$A_D$};
    \node (D1) at (3,0) {$D_1$};
    \node (Y) at (5,0) {$Y$};
\end{scope}

\begin{scope}[>={Stealth[black]},
              every node/.style={fill=white,circle},
              every edge/.style={draw=black,very thick}]
    \path [->] (A) edge[line width=0.85mm] (Ad);
    \path [->] (A) edge[line width=0.85mm] (Ay);
	\path [->] (Ad) edge (D1);
	\path [->] (Ay) edge[bend right] (Y);
    \path [->] (D1) edge (Y);
\end{scope}
\end{tikzpicture}
} }
\subfloat[]{
\scalebox{0.85}{
\begin{tikzpicture}
\begin{scope}[every node/.style={thick,draw=none}]
    \node (A) at (-1.5,0) {$A(G')$};
    \node (Ay) at (0,-0.5) {$A_Y(G')$};
	\node (Ad) at (0,0.5) {$A_D(G')$};
    \node (My) at (2.5,-0.5) {$M_Y(G')$};
	\node (Md) at (2.5,0.5) {$M_D(G')$};
    \node (D1) at (5,0) {$D_1(G')$};
    \node (Y) at (7,0) {$Y(G')$};
\end{scope}

\begin{scope}[>={Stealth[black]},
              every node/.style={fill=white,circle},
              every edge/.style={draw=black,very thick}]
    \path [->] (A) edge[bend left=40] (Md);
    \path [->] (A) edge[bend right=40] (My);
    \path [->] (Ay) edge (My);
    \path [->] (Ad) edge (Md);
	\path [->] (Md) edge (D1);
	\path [->] (My) edge[bend right] (Y);
    \path [->] (D1) edge (Y);
\end{scope}
\end{tikzpicture}
} }
\\ 
\caption{\small The graph in (a) is a simple causal DAG allowing no common cause of $Y$ and $D_1$. The DAG in (b) represents a four-arm trial $G$ assigning two treatments $A_Y$ and $A_D$ consistent with $A_Y$ partial isolation. The graph in (c) is consistent with the modified treatment assumption for this choice of $A_Y$ and $A_D$ via the treatment decomposition assumption \eqref{assumption: Determinsm} from Robins and Richardson \cite{robins2010alternative}, where the bold arrows describe the deterministic relation between $A$, $A_Y$ and $A_D$. The graph in (d) describes a necessary condition for the modified treatment assumption.}
\label{fig: decomposition 1}
\end{figure}


\begin{figure}
    \centering
\subfloat[]{
\scalebox{0.85}{
\begin{tikzpicture}
\begin{scope}[every node/.style={thick,draw=none}]
    \node (A) at (0,0) {$A$};
    \node (D1) at (3,0) {$D_1$};
    \node (D2) at (5,0) {$D_2$};
    \node (Y) at (7,0) {$Y$};
    \node (Z1) at (4,2) {$Z_1$};
    \node (Z0) at (1,2) {$Z_0$};
\end{scope}

\begin{scope}[>={Stealth[black]},
              every node/.style={fill=white,circle},
              every edge/.style={draw=black,very thick}]
 	\path [->] (A) edge[bend left]  (D1);
	\path [->] (A) edge[bend left] (D2);
	\path [->] (A) edge[bend right] (Y);
    \path [->] (Z1) edge (Y);
    \path [->] (D1) edge (Z1);
    \path [->] (D1) edge (D2);
    \path [->] (D2) edge (Y);
    \path [->,>={Stealth[lightgray]}] (Z0) edge[lightgray] (Z1);
    \path [->] (Z1) edge (D2);
    \path [->,>={Stealth[lightgray]}] (Z0) edge[lightgray]  (D1);
    \path [->,>={Stealth[lightgray]}] (Z0) edge[lightgray]  (Y);
    \path [->,>={Stealth[blue]}] (A) edge[blue] (Z1);
\end{scope}
\end{tikzpicture}
}}    
\subfloat[]{
\scalebox{0.85}{
\begin{tikzpicture}
\begin{scope}[every node/.style={thick,draw=none}]
    \node (Ay) at (1,-1) {$A_Y$};
	\node (Ad) at (1,1) {$A_D$};
    \node (D1) at (3,0) {$D_1$};
    \node (D2) at (5,0) {$D_2$};
    \node (Y) at (7,0) {$Y$};
    \node (Z1) at (4,2) {$Z_1$};
    \node (Z0) at (1,2) {$Z_0$};
\end{scope}

\begin{scope}[>={Stealth[black]},
              every node/.style={fill=white,circle},
              every edge/.style={draw=black,very thick}]
	\path [->] (Ad) edge (D1);
	\path [->] (Ad) edge (D2);
	\path [->] (Ay) edge[bend right] (Y);
    \path [->] (Z1) edge (Y);
    \path [->] (D1) edge (Z1);
    \path [->] (D1) edge (D2);
    \path [->] (D2) edge (Y);
    \path [->,>={Stealth[lightgray]}] (Z0) edge[lightgray] (Z1);
    \path [->] (Z1) edge (D2);
    \path [->,>={Stealth[lightgray]}] (Z0) edge[lightgray]  (D1);
    \path [->,>={Stealth[lightgray]}] (Z0) edge[lightgray]  (Y);
\end{scope}
\end{tikzpicture}
}} \\

\subfloat[]{
\scalebox{0.85}{
\begin{tikzpicture}
\begin{scope}[every node/.style={thick,draw=none}]
    \node (Ay) at (1,-1) {$A_Y$};
	\node (Ad) at (1,1) {$A_D$};
    \node (D1) at (3,0) {$D_1$};
    \node (D2) at (5,0) {$D_2$};
    \node (Y) at (7,0) {$Y$};
    \node (Z1) at (4,2) {$Z_1$};
    \node (Z0) at (1,2) {$Z_0$};
\end{scope}

\begin{scope}[>={Stealth[black]},
              every node/.style={fill=white,circle},
              every edge/.style={draw=black,very thick}]
	\path [->] (Ad) edge (D1);
	\path [->] (Ad) edge (D2);
	\path [->] (Ay) edge[bend right] (Y);
    \path [->] (Z1) edge (Y);
    \path [->] (D1) edge (Z1);
    \path [->] (D1) edge (D2);
    \path [->] (D2) edge (Y);
    \path [->,>={Stealth[lightgray]}] (Z0) edge[lightgray] (Z1);
    \path [->] (Z1) edge (D2);
    \path [->,>={Stealth[lightgray]}] (Z0) edge[lightgray]  (D1);
    \path [->,>={Stealth[lightgray]}] (Z0) edge[lightgray]  (Y);
    \path [->,>={Stealth[blue]}] (Ad) edge[blue] (Z1);
\end{scope}
\end{tikzpicture}
} }
\subfloat[]{
\scalebox{0.85}{
\begin{tikzpicture}
\begin{scope}[every node/.style={thick,draw=none}]
    \node (Ay) at (1,-1) {$A_Y$};
	\node (Ad) at (1,1) {$A_D$};
    \node (D1) at (3,0) {$D_1$};
    \node (D2) at (5,0) {$D_2$};
    \node (Y) at (7,0) {$Y$};
    \node (Z1) at (4,2) {$Z_1$};
    \node (Z0) at (1,2) {$Z_0$};
\end{scope}

\begin{scope}[>={Stealth[black]},
              every node/.style={fill=white,circle},
              every edge/.style={draw=black,very thick}]
	\path [->] (Ad) edge (D1);
	\path [->] (Ad) edge (D2);
	\path [->] (Ay) edge[bend right] (Y);
    \path [->] (Z1) edge (Y);
    \path [->] (D1) edge (Z1);
    \path [->] (D1) edge (D2);
    \path [->] (D2) edge (Y);
    \path [->,>={Stealth[lightgray]}] (Z0) edge[lightgray] (Z1);
    \path [->] (Z1) edge (D2);
    \path [->,>={Stealth[lightgray]}] (Z0) edge[lightgray]  (D1);
    \path [->,>={Stealth[lightgray]}] (Z0) edge[lightgray]  (Y);
    \path [->,>={Stealth[blue]}] (Ay) edge[blue, bend left] (Z1);
\end{scope}
\end{tikzpicture}
}}
\caption{\small The graph in (a) is a classical DAG with no assumption about modified treatments. The graphs in (b)-(d) are DAGs describing four-arm trials, i.e. "$G$"s, that illustrate the isolation assumptions for modified treatments $A_Y$ and $A_D$, where we have omitted $(G)$ in each node to avoid clutter. Full isolation holds in (b), only $A_Y$ partial isolation holds in (c) and only $A_D$ partial isolation holds in (d).}
\label{fig: isolation conditions}
\end{figure}
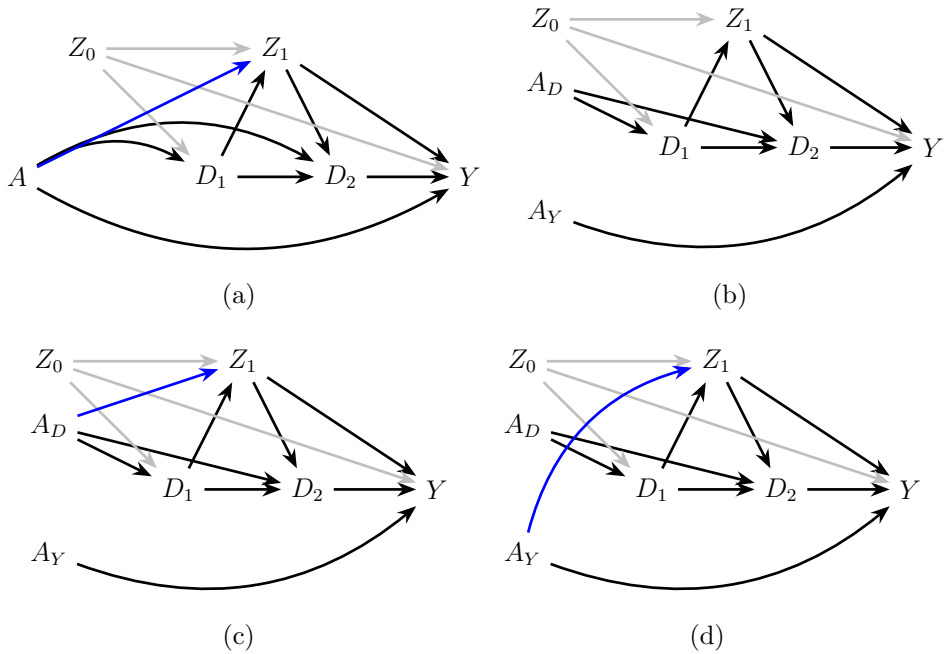

\begin{figure}
\centering
\subfloat[]{
\scalebox{0.85}{
\begin{tikzpicture}
\tikzset{line width=1.5pt, outer sep=0pt,
ell/.style={draw,fill=white, inner sep=2pt,
line width=1.5pt},
swig vsplit={gap=5pt,
inner line width right=0.5pt}};
\node[name=Ay,shape=swig vsplit] at (0,-1){
\nodepart{left}{$A_Y$}
\nodepart{right}{$a_Y$} };
\node[name=Ad, 
shape=swig vsplit] at (0,1) {
\nodepart{left}{$A_D$}
\nodepart{right}{$a_D$} };
\node[name=D1,ell,  shape=ellipse] at (3,0) {$D^{a_D}_1$}  ;
    \node[name=D2,ell,  shape=ellipse]  at (5.5,0) {$D^{a_D}_2$};
    \node[name=Y,ell,  shape=ellipse] (Y) at (8,0) {$Y^{a_Y,a_D}$};
    \node[name=Z1,ell,  shape=ellipse] at (5,2.5) {$Z^{a_D}_1$};
    \node[name=Z0,ell,  shape=ellipse] at (1,2.5) {$Z_0$};
\begin{scope}[>={Stealth[black]},
              every edge/.style={draw=black,very thick}]
	\path [->] (Ad) edge (D1);
	\path [->] (Ad) edge[bend left=10] (D2);
	\path [->] (Ay) edge[bend right] (Y);
    \path [->] (Z1) edge (Y);
    \path [->] (D1) edge (Z1);
    \path [->] (D1) edge (D2);
    \path [->] (D2) edge (Y);
    \path [->,>={Stealth[lightgray]}] (Z0) edge[lightgray] (Z1);
    \path [->] (Z1) edge (D2);
    \path [->,>={Stealth[lightgray]}] (Z0) edge[lightgray]  (D1);
    \path [->,>={Stealth[lightgray]}] (Z0) edge[lightgray]  (Y);
    \path [->,>={Stealth[blue]}] (Ad) edge[blue] (Z1);
\end{scope}
\end{tikzpicture}
}
} \\
\subfloat[]{
\scalebox{0.85}{
\begin{tikzpicture}
\tikzset{line width=1.5pt, outer sep=0pt,
ell/.style={draw,fill=white, inner sep=2pt,
line width=1.5pt},
swig vsplit={gap=5pt,
inner line width right=0.5pt}};
\node[name=Ay,shape=swig vsplit] at (0,-1){
\nodepart{left}{$A_Y$}
\nodepart{right}{$a_Y$} };
\node[name=Ad, 
shape=swig vsplit] at (0,1) {
\nodepart{left}{$A_D$}
\nodepart{right}{$a_D$} };
\node[name=D1,ell,  shape=ellipse] at (3,0) {$D^{a_D}_1$}  ;
    \node[name=D2,ell,  shape=ellipse]  at (5.5,0) {$D^{a_Y,a_D}_2$};
    \node[name=Y,ell,  shape=ellipse] (Y) at (8,0) {$Y^{a_Y,a_D}$};
    \node[name=Z1,ell,  shape=ellipse] at (5,2.5) {$Z^{a_Y,a_D}_1$};
    \node[name=Z0,ell,  shape=ellipse] at (1,2.5) {$Z_0$};
\begin{scope}[>={Stealth[black]},
              every edge/.style={draw=black,very thick}]
	\path [->] (Ad) edge (D1);
	\path [->] (Ad) edge[bend left=10] (D2);
	\path [->] (Ay) edge[bend right] (Y);
    \path [->] (Z1) edge (Y);
    \path [->] (D1) edge (Z1);
    \path [->] (D1) edge (D2);
    \path [->] (D2) edge (Y);
    \path [->,>={Stealth[lightgray]}] (Z0) edge[lightgray] (Z1);
    \path [->] (Z1) edge (D2);
    \path [->,>={Stealth[lightgray]}] (Z0) edge[lightgray]  (D1);
    \path [->,>={Stealth[lightgray]}] (Z0) edge[lightgray]  (Y);
    \path [->,>={Stealth[blue]}] (Ay) edge[blue] (Z1);
\end{scope}
\end{tikzpicture}
}
}
    \caption{\small Single World Intervention Graphs (SWIGs). The SWIG in (a) is a transformation of the causal DAG in Figure \ref{fig: isolation conditions}c under an intervention $a_Y,a_D$, where $A_Y$ partial isolation holds. The SWIG in (b) is a corresponding transformation of the causal DAG in Figure \ref{fig: isolation conditions}d, where $A_Y$ partial isolation fails. Unlike in the graph (a), we cannot remove $a_Y$ from the superscripts of $Z_1
   ^{a_Y,a_D}$ $D^{a_Y,a_D}_2$ in (b) by minimal labelling \cite{richardson2013single}.}
    \label{fig: swig minimal labelling}
\end{figure}
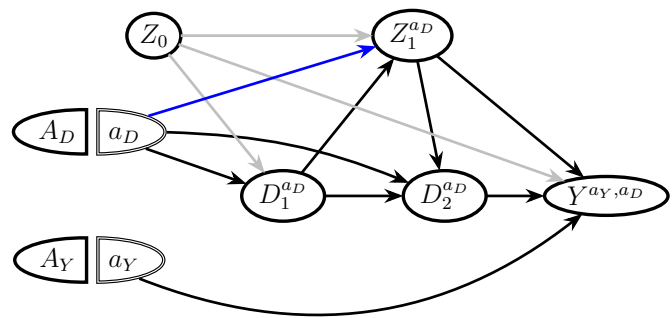
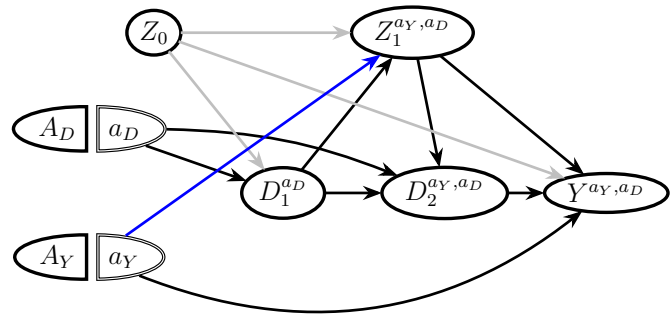

\begin{figure}
    \centering
\subfloat[]{
\begin{tikzpicture}
\begin{scope}[every node/.style={thick,draw=none}]
    \node (Ay) at (1,-1) {$A_Y$};
	\node (Ad) at (1,1) {$A_D$};
    \node (D1) at (3,0) {$D_1$};
    \node (D2) at (5,0) {$D_2$};
    \node (Y) at (7,0) {$Y$};
    \node (L1) at (4,2) {$L_1$};
    \node[lightgray] (ULD) at (2,2) {$U_{L,D}$};
\end{scope}

\begin{scope}[>={Stealth[black]},
              every node/.style={fill=white,circle},
              every edge/.style={draw=black,very thick}]
	\path [->] (Ad) edge (D1);
	\path [->] (Ad) edge (D2);
	\path [->] (Ay) edge[bend right] (Y);
    \path [->] (L1) edge (Y);
    \path [->] (D1) edge (L1);
    \path [->] (D1) edge (D2);
    \path [->] (D2) edge (Y);
    \path [->] (L1) edge (D2);
    \path [->,>={Stealth[lightgray]}] (ULD) edge[lightgray] (L1);
    \path [->,>={Stealth[lightgray]}] (ULD) edge[lightgray] (D2);
\end{scope}
\end{tikzpicture}
}
\subfloat[]{
\begin{tikzpicture}
\begin{scope}[every node/.style={thick,draw=none}]
    \node (Ay) at (1,-1) {$A_Y$};
	\node (Ad) at (1,1) {$A_D$};
    \node (D1) at (3,0) {$D_1$};
    \node (D2) at (5,0) {$D_2$};
    \node (Y) at (7,0) {$Y$};
    \node (L1) at (4,2) {$L_1$};
    \node[lightgray] (ULY) at (6,2) {$U_{L,Y}$};
\end{scope}

\begin{scope}[>={Stealth[black]},
              every node/.style={fill=white,circle},
              every edge/.style={draw=black,very thick}]
	\path [->] (Ad) edge (D1);
	\path [->] (Ad) edge (D2);
	\path [->] (Ay) edge[bend right] (Y);
    \path [->] (L1) edge (Y);
    \path [->] (D1) edge (L1);
    \path [->] (D1) edge (D2);
    \path [->] (D2) edge (Y);
    \path [->] (L1) edge (D2);
    \path [->,>={Stealth[lightgray]}] (ULY) edge[lightgray] (L1);
    \path [->,>={Stealth[lightgray]}] (ULY) edge[lightgray] (Y);
\end{scope}
\end{tikzpicture}
}  \\
\subfloat[]{
\begin{tikzpicture}
\begin{scope}[every node/.style={thick,draw=none}]
    \node (Ay) at (1,-1) {$A_Y$};
	\node (Ad) at (1,1) {$A_D$};
    \node (D1) at (3,0) {$D_1$};
    \node (D2) at (5,0) {$D_2$};
    \node (Y) at (7,0) {$Y$};
    \node (L1) at (4,2) {$L_1$};
    \node[lightgray] (ULD) at (2,2) {$U_{L,D}$};
\end{scope}
\begin{scope}[>={Stealth[black]},
              every node/.style={fill=white,circle},
              every edge/.style={draw=black,very thick}]
	\path [->] (Ad) edge (D1);
	\path [->] (Ad) edge (D2);
	\path [->] (Ay) edge[bend right] (Y);
    \path [->] (L1) edge (Y);
    \path [->] (D1) edge (L1);
    \path [->] (D1) edge (D2);
    \path [->] (D2) edge (Y);
    \path [->] (L1) edge (D2);
    \path [->,>={Stealth[lightgray]}] (ULD) edge[lightgray] (L1);
    \path [->,>={Stealth[lightgray]}] (ULD) edge[lightgray] (D2);
    \path [->,>={Stealth[blue]}] (Ad) edge[blue] (L1);
\end{scope}
\end{tikzpicture}
} 
\subfloat[]{
\begin{tikzpicture}
\begin{scope}[every node/.style={thick,draw=none}]
    \node (Ay) at (1,-1) {$A_Y$};
	\node (Ad) at (1,1) {$A_D$};
    \node (D1) at (3,0) {$D_1$};
    \node (D2) at (5,0) {$D_2$};
    \node (Y) at (7,0) {$Y$};
    \node (L1) at (4,2) {$L_1$};
    \node[lightgray] (ULY) at (6,2) {$U_{L,Y}$};
\end{scope}
\begin{scope}[>={Stealth[black]},
              every node/.style={fill=white,circle},
              every edge/.style={draw=black,very thick}]
	\path [->] (Ad) edge (D1);
	\path [->] (Ad) edge (D2);
	\path [->] (Ay) edge[bend right] (Y);
    \path [->] (L1) edge (Y);
    \path [->] (D1) edge (L1);
    \path [->] (D1) edge (D2);
    \path [->] (D2) edge (Y);
    \path [->] (L1) edge (D2);
    \path [->,>={Stealth[lightgray]}] (ULY) edge[lightgray] (L1);
    \path [->,>={Stealth[lightgray]}] (ULY) edge[lightgray] (Y);
    \path [->,>={Stealth[blue]}] (Ad) edge[blue] (L1);
\end{scope}
\end{tikzpicture}
} \\
\subfloat[]{
\begin{tikzpicture}
\begin{scope}[every node/.style={thick,draw=none}]
    \node (Ay) at (1,-1) {$A_Y$};
	\node (Ad) at (1,1) {$A_D$};
    \node (D1) at (3,0) {$D_1$};
    \node (D2) at (5,0) {$D_2$};
    \node (Y) at (7,0) {$Y$};
    \node (L1) at (4,2) {$L_1$};
    \node[lightgray] (ULY) at (6,2) {$U_{L,Y}$};
    \node[lightgray] (ULD) at (2,2) {$U_{L,D}$};
\end{scope}
\begin{scope}[>={Stealth[black]},
              every node/.style={fill=white,circle},
              every edge/.style={draw=black,very thick}]
	\path [->] (Ad) edge (D1);
	\path [->] (Ad) edge (D2);
	\path [->] (Ay) edge[bend right] (Y);
    \path [->] (L1) edge (Y);
    \path [->] (D1) edge (L1);
    \path [->] (D1) edge (D2);
    \path [->] (D2) edge (Y);
    \path [->] (L1) edge (D2);
    \path [->,>={Stealth[lightgray]}] (ULY) edge[lightgray] (L1);
    \path [->,>={Stealth[lightgray]}] (ULY) edge[lightgray] (Y);
    \path [->,>={Stealth[lightgray]}] (ULD) edge[lightgray] (L1);
    \path [->,>={Stealth[lightgray]}] (ULD) edge[lightgray] (D2);
\end{scope}
\end{tikzpicture}
}
\subfloat[]{
\begin{tikzpicture}
\begin{scope}[every node/.style={thick,draw=none}]
    \node (Ay) at (1,-1) {$A_Y$};
	\node (Ad) at (1,1) {$A_D$};
    \node (D1) at (3,0) {$D_1$};
    \node (D2) at (5,0) {$D_2$};
    \node (Y) at (7,0) {$Y$};
    \node (L1) at (4,2) {$L_1$};
\end{scope}
\begin{scope}[>={Stealth[black]},
              every node/.style={fill=white,circle},
              every edge/.style={draw=black,very thick}]
	\path [->] (Ad) edge (D1);
	\path [->] (Ad) edge (D2);
	\path [->] (Ay) edge[bend right] (Y);
    \path [->] (L1) edge (Y);
    \path [->] (D1) edge (L1);
    \path [->] (D1) edge (D2);
    \path [->] (D2) edge (Y);
    \path [->] (L1) edge (D2);
    \path [->,>={Stealth[blue]}] (Ay) edge[blue, bend left=10] (L1);
    \path [->,>={Stealth[blue]}] (Ad) edge[blue] (L1);
\end{scope}
\end{tikzpicture}
}

\caption{The dismissible component conditions can be studied in causal graphs. Similar to Figure \ref{fig: isolation conditions}, we have omitted the string $(G)$ in each node to avoid clutter. In (a)-(c), the dismissible component conditions hold. In (d)-(e), dismissible component condition \eqref{ass: delta 1} is violated. In (f), dismissible component condition \eqref{ass: delta 3} is violated. }
\label{fig: dismissible component conditions}
\end{figure}
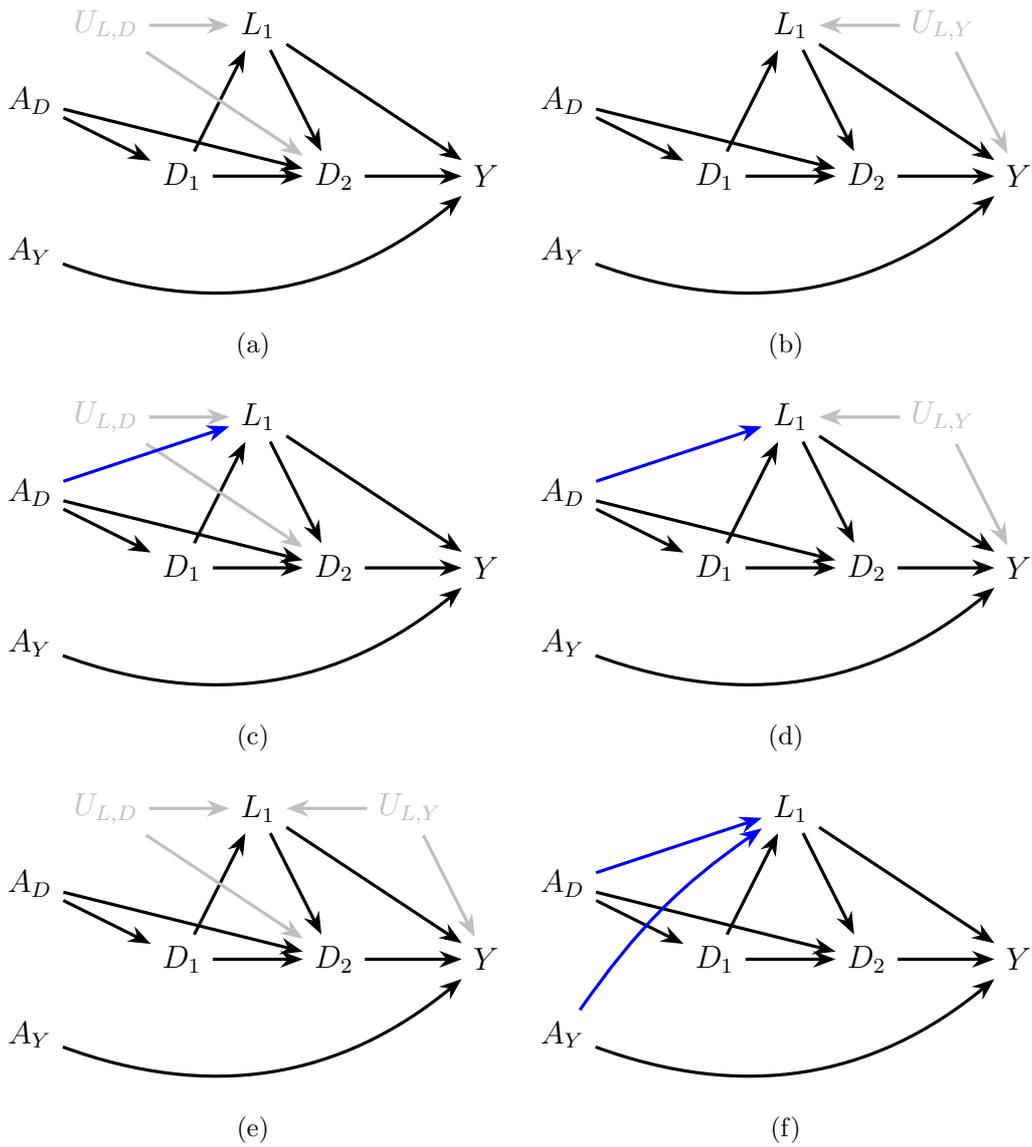

\clearpage
\appendix

\section{Proof of identification formula}
\label{app sec: proof of id}

Here we provide a proof for the identification formula of $\mathbb{E} ( Y^{a_Y,a_D}\mid
D^{a_Y,a_D}_{K+1}=0 )$, which is sufficient for identification of the conditional separable effects in settings with loss to follow-up (censoring). This identification formula co-incide with the identification formula in Malinsky et al \cite{malinsky2019potential} for path specific effects (See also Robins et al \cite{robins2020interventionist}). Denote $C_{k+1}$ an indicator of censoring by $k+1$, $C_0\equiv0$ and assume topological ordering $(C_{k+1},D_{k+1},L_{k+1}$) in each $k=0,\ldots,K$. Consider the following more general exchangeability, consistency and positivity conditions which coincide with those given in Section \ref{sec: identifiability conditions} of the main text when loss to follow-up is absent.  For each $a\in\{0,1\}$:
\begin{enumerate}
\item[1. ] Exchangeability: 
\begin{align}
& Y^{a, \overline{c}=0},\underline{D}_{1}^{a, \overline{c}=0},\underline{L}^{a, \overline{c}=0}_{1} \independent A \mid L_{0} \label{ass: E1 app} \\
&  Y^{a, \overline{c}=0}, \underline{D}^{a, \overline{c}=0}_{k+1}, \underline{L}^{a, \overline{c}=0}_{k+1}  \independent C_{k+1} \mid D_k = \overline{C}_k = 0, \overline{L}_k, A. \label{ass: E2 app}
\end{align}
Condition \eqref{ass: E1 app} holds if $A$ is randomly assigned at baseline, possibly conditional on $L_0$. Condition \eqref{ass: E2 app}  states that losses to follow-up are independent of future counterfactual events, conditional on the measured past. This assumption does not hold by design in a trial where $A$ is randomly assigned, because losses to follow-up are not randomly assigned in practice.
\item[2.] Positivity:  
\begin{align}
& f_{L_0}(l_0)>0\implies   \nonumber \\
& \quad \Pr (A=a\mid  L_0=l_0)>0, \label{eq: positivity 1 app}  \\ 
& f_{\overline{L}_k,C_{k+1},D_{k+1}}(\overline{l}_k,0,0)\neq 0  \text{ for } a\in\{0,1\} \implies \nonumber \\  
& \quad  \Pr(A=a|D_{k+1}=0,\overline{C}_{k+1}=0,\overline{L}_k=\overline{l}_k)>0 \label{eq: positivity 2 app} \\
& f(A = a,D_k=0,\overline{C}_k=0,\overline{L}_k=\overline{l}_k) > 0 \implies \nonumber \\
&\quad \Pr(C_{k+1}=0\mid D_k=0,\overline{C}_k=0,\overline{L}_k=\overline{l}_k,A=a) > 0. \label{eq: positivity 3 app} 
\end{align}
Here, \eqref{eq: positivity 1 app} and \eqref{eq: positivity 2 app} are analogous to the conditions in the main text. Condition \eqref{eq: positivity 3 app} ensures that for any possible history of treatment assignments and covariates among those who are event-free and uncensored at $k$, some individuals will remain uncensored at $k+1$. Condition \eqref{eq: positivity 3 app} is only required when loss to follow-up is present. 
\item[3.] Consistency: 
\begin{align}
  & \text{if } A=a \text{ and } \overline{C}_{k+1} = 0, \nonumber \\
  & \text{then } \overline{Y} = \overline{Y}^{a,  \overline{c}=0}, \overline{D}_{k+1} = \overline{D}^{a,  \overline{c}=0}_{k+1} \text{ and } \overline{L}_{k+1} = \overline{L}^{a,\overline{c}=0}_{k+1}.
  \label{ass: consistency cens}
\end{align}
Consistency is satisfied any individual who has data history consistent with the intervention under a counterfactual scenario, would have an observed outcome that is equal to the counterfactual outcome. 

Also consider the following more general version of the dismissible component conditions where $(L_{A_Y,k+1},L_{A_D,k+1})$ is some partition of $L_{k+1}$:
\item[4.] Dismissible component conditions: 
\begin{align}
 & Y(G) \independent A_D(G) \mid A_Y(G),D_{K+1}(G)=0, \overline{L}_K(G), \label{ass: delta 1 app} \\
 & D_{k+1}(G) \independent A_Y(G) \mid A_D(G), D_{k}(G)=0, \overline{L}_k(G), \label{ass: delta 2 app} \\
 & L_{A_Y,k+1}(G) \independent A_D(G) \mid A_Y(G), D_{k+1}(G)=0, \overline{L}_{k}(G),L_{A_D,k+1}(G),  \label{ass: delta 3a app} \\
 & L_{A_D,k+1}(G) \independent A_Y(G) \mid A_D(G), D_{k+1}(G)=0. \overline{L}_{k}(G)  \label{ass: delta 3b app} 
\end{align}
Following Appendix C of \cite{stensrud2019generalized}, these dismissible component conditions and $A_Y$ partial isolation give the dismissible component conditions of Section \ref{sec: identifiability conditions} in the main text; conditions \eqref{ass: delta 1}-\eqref{ass: delta 3} are equivalent to \eqref{ass: delta 2 app}-\eqref{ass: delta 3b app} under the partition $L_{A_Y,k+1}=L_{k+1}$ and $L_{A_D,k+1}=\emptyset$ .
\end{enumerate}

\begin{theorem}
\label{theorem g formula}
Under the modified treatment assumption, if conditions \eqref{ass: E1 app}-\eqref{ass: delta 3b app} hold, $\mathbb{E}  ( Y^{a_Y,a_D,\overline{c}=0} \mid 
D^{a_Y,a_D,\overline{c}=0}_ {K+1}=0 )$ is identified by
\begin{align}
   &\mathbb{E}  ( Y^{a_Y,a_D,\overline{c}=0} \mid 
D^{a_Y,a_D,\overline{c}=0}_ {K+1}=0) \nonumber \\
= & \sum_{\overline{l}_K} \mathbb{E}  ( Y \mid D_{K+1}=C_{K+1}=0,  \overline{L}_{K}=\overline{l}_{K}, A=a_Y) \nonumber \\
&  \prod_{s=0}^{K} \bigg[ \Pr ( D_{s+1}=0 \mid \overline{L}_{s} = \overline{l}_{s},D_{s}=C_{s+1}=0, A=a_D) \nonumber \\
& \times \Pr ( L_{A_Y,s}=l_{A_Y,s} \mid L_{A_D,s} =  l_{A_D,s}, \overline{L}_{s-1} = \overline{l}_{s-1},D_{s}=C_{s}=0, A =a_Y) \nonumber \\
& \times \Pr ( L_{A_D,s}=l_{A_D,s} \mid \overline{L}_{s-1} = \overline{l}_{s-1},D_{s}=C_{s}=0, A=a_D) \bigg] \nonumber \\
\times & \Bigg[ \sum_{\overline{l}_K} \prod_{s=0}^{K}  \Pr ( D_{s+1}=0 \mid \overline{L}_{s} = \overline{l}_{s},D_{s}=C_{s}=0, A=a_D) \nonumber \\
& \times \Pr ( L_{A_Y,s}=l_{A_Y,s} \mid  L_{A_D,s}=l_{A_D,s}, \overline{L}_{s-1} = \overline{l}_{s-1},D_{s}=C_{s}=0, A =a_Y) \nonumber \\
& \times \Pr ( L_{A_D,s}=l_{A_D,s} \mid \overline{L}_{s-1} =\overline{l}_{s-1},D_{s}=C_{s}=0, A=a_D) \Bigg]^{-1},
\label{eq: id g-formula}
\end{align} 
which reduces to the identification formula \eqref{eq: id formula no cens} when $A_Y$ partial isolation holds, there is no censoring and $A \independent L_0$. 
\label{theorem: identification g formula}
\end{theorem}

\begin{proof}
Assume that conditions \eqref{ass: E1 app}-\eqref{ass: delta 3b app} hold. Using laws of probability, 
\begin{align}
& \mathbb{E}  ( Y^{a_Y,a_D,\overline{c}=0} \mid D^{a_Y,a_D,\overline{c}=0}_ {K+1}=0 ) \nonumber \\
= &  \sum_{\overline{l}_{K}} \mathbb{E}  ( Y^{a_Y,a_D,\overline{c}=0} \mid 
D^{a_Y,a_D,\overline{c}=0}_ {K+1}=0, \overline{L}^{a_Y,a_D,\overline{c}=0}_ {K}=\overline{l}_ {K} ) \Pr(\overline{L}^{a_Y,a_D,\overline{c}=0}_{K}=\overline{l}_{K} \mid 
D^{a_Y,a_D,\overline{c}=0}_ {K+1}=0) \nonumber \\
= &  \sum_{\overline{l}_ {K}} \mathbb{E}  ( Y^{a_Y,a_D,\overline{c}=0} \mid 
D^{a_Y,a_D,\overline{c}=0}_ {K+1}=0,  \overline{L}^{a_Y,a_D,\overline{c}=0}_{K}=\overline{l}_{K} ) \frac{\Pr(\overline{L}^{a_Y,a_D,\overline{c}=0}_{K}=\overline{l}_{K}, D^{a_Y,a_D,\overline{c}=0}_ {K+1}=0)}{\Pr(D^{a_Y,a_D,\overline{c}=0}_ {K+1}=0)} \nonumber \\
= & \frac{\sum_{\overline{l}_ {K}} \mathbb{E}  ( Y^{a_Y,a_D,\overline{c}=0} \mid 
D^{a_Y,a_D,\overline{c}=0}_ {K+1}=0,  \overline{L}^{a_Y,a_D,\overline{c}=0}_{K}=\overline{l}_{K} ) \Pr(\overline{L}^{a_Y,a_D,\overline{c}=0}_{K}=\overline{l}_{K}, D^{a_Y,a_D,\overline{c}=0}_ {K+1}=0)}{\Pr(D^{a_Y,a_D,\overline{c}=0}_ {K+1}=0)} \nonumber \\
\label{eq: counterfactual rep}
\end{align}
Our task now remains to express each term in \eqref{eq: counterfactual rep} as a function of observed data under the stated identifying conditions above. We begin with the term $\Pr(\overline{L}^{a_Y,a_D,\overline{c}=0}_{K}=\overline{l}_{K}, D^{a_Y,a_D,\overline{c}=0}_ {K+1}=0)$.  By probability laws and a partitioning $L_{k+1}=(L_{A_Y,k+1},L_{A_D,k+1})$, $k=0,\ldots,K$, we have:
\begin{align}
& \Pr(\overline{L}^{a_Y,a_D,\overline{c}=0}_{K}=\overline{l}_{K}, D^{a_Y,a_D,\overline{c}=0}_ {K+1}=0) \nonumber \\
= & \Pr( D^{a_Y,a_D,\overline{c}=0}_ {K+1}=0 \mid \overline{L}^{a_Y,a_D,\overline{c}=0}_ {K}=\overline{l}_{K},D^{a_Y,a_D,\overline{c}=0}_ {K}=0) \nonumber \\
& \times \Pr(L_{A_Y,K}^{a_Y,a_D,\overline{c}=0}=l_{A_Y,K} \mid  D^{a_Y,a_D,\overline{c}=0}_ {K}=0, \overline{L}^{a_Y,a_D,\overline{c}=0}_ {K-1}=\overline{l}_{K-1}, L_{A_D,K}^{a_Y,a_D\overline{c}=0}=l_{A_D,K}) \nonumber \\
& \times \Pr(L_{A_D,K}^{a_Y,a_D,\overline{c}=0} = l_{A_D,K} \mid D^{a_Y,a_D,\overline{c}=0}_{K}=0, \overline{L}^{a_Y,a_D,\overline{c}=0}_{K-1}=\overline{l}_{K-1} ) \nonumber \\
& \times \Pr(D^{a_Y,a_D,\overline{c}=0}_ {K}=0, \overline{L}^{a_Y,a_D,\overline{c}=0}_ {K-1} = \overline{l}_{k-1}) \nonumber \\
& \text{ Arguing iteratively:} \nonumber \\
= & \prod_{s=0}^{K} \Pr( D^{a_Y,a_D,\overline{c}=0}_{s+1}=0 \mid \overline{L}^{a_Y,a_D,\overline{c}=0}_{s}=\overline{l}_s,D^{a_Y,a_D,\overline{c}=0}_{s}=0) \nonumber \\
 & \times \Pr(L_{A_Y,s}^{a_Y,a_D,\overline{c}=0}=l_{A_Y,s} \mid  D^{a_Y,a_D,\overline{c}=0}_{s}=0, L^{a_Y,a_D,\overline{c}=0}_{s-1}=l_{s-1},L_{A_D,s}^{a_Y,a_D,\overline{c}=0}=l_{A_D,s}) \nonumber \\ 
 &\times  \Pr(L_{A_D,s}^{a_Y,a_D,\overline{c}=0}=l_{A_D,s}\mid D^{a_Y,a_D,\overline{c}=0}_{s}=0, \overline{L}^{a_Y,a_D,\overline{c}=0}_{s-1} = \overline{l}_{s-1} ) \nonumber \\
& \text{Using the dismissible component conditions$^*$ and the modified treatment assumption$^{**}$:} \nonumber \\
= & \prod_{s=0}^{K}  \Pr( D^{a=a_D,\overline{c}=0}_{s+1}=0 \mid \overline{L}^{a=a_D,\overline{c}=0}_{s}=\overline{l}_s,D^{a=a_D,\overline{c}=0}_{s}=0) \nonumber \\
& \times \Pr(L^{a=a_Y,\overline{c}=0}_{A_Y, s} = l_{A_Y, s} \mid  D^{a=a_Y,\overline{c}=0}_{s}=0, \overline{L}^{a=a_Y,\overline{c}=0}_{s-1}=\overline{l}_{s-1}, L^{a=a_Y,\overline{c}=0}_{A_D, s}=l_{A_D, s} ) \nonumber \\
& \times \Pr(L^{a=a_D,\overline{c}=0}_{A_D, s}=l_{A_D, s} \mid D^{a=a_D,\overline{c}=0}_{s}=0, \overline{L}^{a=a_D,\overline{c}=0}_{s-1}=\overline{l}_{s-1} ) . \nonumber \\ \label{eq: joint L,D}
\end{align}

\underline{Remark on last equality}
\\
* The dismissible component conditions, which are defined as independencies, imply the following equalities of conditional hazards (See Lemma 1, in Stensrud et al \cite{stensrud2019generalized}),
\begin{align*}
 & \Pr(Y^{a_Y,a_D=0,\bar{c}= 0} = 1 \mid   D^{a_Y,a_D=0,\bar{c}= 0}_{k+1} = 0,\bar{L}^{a_Y,a_D=0,\bar{c}= 0}_{k} = \bar{l}_{k}) \\ 
=&\Pr(Y^{a_Y,a_D=1,\bar{c}= 0} = 1 \mid  D^{a_Y,a_D=1,\bar{c}=0}_{k+1} = 0 , \bar{L}^{a_Y,a_D=1,\bar{c}= 0}_{k} = \bar{l}_{k} ), \nonumber \\
 & \Pr(D^{a_Y=0,a_D,\bar{c}= 0}_{k+1} = 1 \mid   D^{a_Y=0,a_D,\bar{c}= 0}_k = 0, \bar{L}^{a_Y=0,a_D,\bar{c}= 0}_{k} = \bar{l}_{k}) \\ 
=&\Pr(D^{a_Y=1,a_D,\bar{c}= 0}_{k+1} = 1 \mid   D^{a_Y=1,a_D,\bar{c}=0}_k = 0 ,  \bar{L}^{a_Y=1,a_D,\bar{c}= 0}_{k} = \bar{l}_{k} ), \nonumber \\ 
   & \Pr(L^{a_Y,a_D=1,\bar{c}= 0}_{A_Y,k+1} = l_{A_Y,k+1} \mid   D^{a_Y,a_D=1,\bar{c}= 0}_{k+1} = 0, 
   \bar{L}^{a_Y,a_D=1,\bar{c}= 0}_{k} = \bar{l}_{k},L^{a_Y,a_D=1,\bar{c}= 0}_{A_D,k+1} = l_{A_D,k+1}) \nonumber \\
   =& \Pr (L^{a_Y,a_D=0,\bar{c}= 0}_{A_Y,k+1} = l_{A_Y,k+1} \mid  D^{a_Y,a_D=0,\bar{c}= 0}_{k+1} = 0, \bar{L}^{a_Y,a_D=0,\bar{c}= 0}_{k} = \bar{l}_{k},L^{a_Y,a_D=0,\bar{c}= 0}_{A_D,k+1} = l_{A_D,k+1}), \nonumber \\
  & \Pr(L^{a_Y=0,a_D,\bar{c}= 0}_{A_D,k+1} = l_{A_D,k+1} \mid   D^{a_Y=0,a_D,\bar{c}= 0}_{k+1} = 0, \bar{L}^{a_Y=0,a_D,\bar{c}= 0}_{k} = \bar{l}_{k}) \\ 
   =& \Pr(L^{a_Y=1,a_D,\bar{c}= 0}_{A_D,k+1} = l_{A_D,k+1} \mid   D^{a_Y=1,a_D,\bar{c}= 0}_{k+1} = 0, \bar{L}^{a_Y=1,a_D,\bar{c}= 0}_{k} = \bar{l}_{k}). \nonumber
\end{align*}

$**$ Note that the modified treatment assumption is defined with respect to $Z_s$, not $L_s$. Thus, we can only apply this condition to $L_s$ such that $L_s \subset Z_k$. Clearly, it is practically sensible to select $L_s$ such that $L_s \subset Z_k$. However, suppose that the investigator selected $L_s$ such that $L'_s \subset L_s$ and $L'_s \not\subset Z_s$ and the dismissible component conditions hold. Then, the last equality still holds because, by definition of $Z_s$, $L'_s$ does not exert effects on either $Y$ or $D_k$, so it can be removed (and re-included) from all the conditioning sets. 
\\ \underline{End of Remark.}
\\

Next consider the term  $\Pr (D^{a_Y,a_D,\overline{c}=0}_ {K+1}=0)$ in the denominator of \eqref{eq: counterfactual rep}.  Invoking probability rules, the dismissible component conditions and the modified treatment assumption, we have:
\begin{align}
\Pr & (D^{a_Y,a_D,\overline{c}=0}_ {K+1}=0) \nonumber \\
= & \sum_{\overline{l}_ {K}} \prod_{s=0}^{K} \Pr( D^{a_Y,a_D,\overline{c}=0}_{s+1}=0 \mid \overline{L}^{a_Y,a_D,\overline{c}=0}_{s} = \overline{l}_s,D^{a_Y,a_D,\overline{c}=0}_{s}=0) \nonumber \\
& \times \Pr(L_{A_Y,s}^{a_Y,a_D,\overline{c}=0}=l_{A_Y,s} \mid  D^{a_Y,a_D,\overline{c}=0}_{s}=0, \overline{L}^{a_Y,a_D,\overline{c}=0}_{s-1} = \overline{l}_{s-1},L_{A_D,s}^{a_Y,a_D,\overline{c}=0} = l_{A_D,s}) \nonumber \\ 
& \times \Pr(L_{A_D,s}^{a_Y,a_D,\overline{c}=0} = l_{A_D,s} \mid D^{a_Y,a_D,\overline{c}=0}_{s}=0, \overline{L}^{a_Y,a_D,\overline{c}=0}_{s-1} = \overline{l}_{s-1} ) \nonumber \\
= & \sum_{\overline{l}_{K}}  \prod_{s=0}^{K} \Pr( D^{a=a_D,\overline{c}=0}_{s+1}=0 \mid \overline{L}^{a=a_D,\overline{c}=0}_{s}=\overline{l}_s=\overline{l}_s,D^{a=a_D,\overline{c}=0}_{s}=0) \nonumber \\
& \times \Pr(L^{a=a_Y,\overline{c}=0}_{A_Y, s} = l_{A_Y, s} \mid  D^{a=a_Y,\overline{c}=0}_{s}=0, \overline{L}^{a=a_Y,\overline{c}=0}_{s-1} = \overline{l}_{s-1}, L^{a=a_Y,\overline{c}=0}_{A_D, s} = l_{A_D, s} ) \nonumber \\
& \times \Pr(L^{a=a_D,\overline{c}=0}_{A_D, s} = l_{A_D, s}   \mid D^{a=a_D,\overline{c}=0}_{s}=0, \overline{L}^{a=a_D,\overline{c}=0}_{s-1} = \overline{l}_{s-1}). 
\label{eq: marginal D}
\end{align}
Note that for $s \geq 0$, invoking exchangeability and positivity,
\begin{align*}
\Pr & ( D^{a=a_D,\overline{c}=0}_{s+1}=0 \mid \overline{L}^{a=a_D,\overline{c}=0}_{s}=\overline{l}_s,D^{a=a_D,\overline{c}=0}_{s}=0) \\
= & \Pr( D^{a=a_D,\overline{c}=0}_{s+1}=0 \mid \overline{L}^{a=a_D,\overline{c}=0}_{s}=\overline{l}_s,D^{a=a_D,\overline{c}=0}_{s}=0, D_{0}=\overline{C}_{0}=0,L_0) \text{ } \\
= & \frac{\Pr( D^{a=a_D,\overline{c}=0}_{s+1}=0,  \overline{L}^{a=a_D,\overline{c}=0}_{s}=\overline{l}_s \mid D_{0}=\overline{C}_{0}=0,L_0, A = a_D) }{\Pr( D^{a=a_D,\overline{c}=0}_{s}=0,  \overline{L}^{a=a_D,\overline{c}=0}_{s}=\overline{l}_s \mid D_{0}=\overline{C}_{0}=0,L_0, A = a_D) } \\ 
= &  \frac{\Pr( D^{a=a_D,\overline{c}=0}_{s+1}=0,  \overline{L}^{a=a_D,\overline{c}=0}_{s}=\overline{l}_s \mid D_{0}=\overline{C}_{1}=0,L_0, A = a_D) }{\Pr( D^{a=a_D,\overline{c}=0}_{s}=0,  \overline{L}^{a=a_D,\overline{c}=0}_{s}=\overline{l}_s \mid D_{0}=\overline{C}_{1}=0,L_0, A = a_D) } \\ 
= & \Pr ( D^{a=a_D,\overline{c}=0}_{s+1}=0 \mid  \overline{L}^{a=a_D,\overline{c}=0}_{s}=\overline{l}_s, D^{a=a_D,\overline{c}=0}_{s}=0, D_{0}=\overline{C}_{1}=0,L_0=l_0, A = a_D)  \\
& \text{ Using consistency and positivity:} \\
= & \Pr ( D^{a=a_D,\overline{c}=0}_{s+1}=0 \mid  \overline{L}^{a=a_D,\overline{c}=0}_{s}=\overline{l}_s, D^{a=a_D,\overline{c}=0}_{s}=0, D_{1}=\overline{C}_{1}=0,\overline{L}_1=\overline{l}_1, A = a_D), \\
\end{align*}
and arguing iteratively using exchangeability, consistency and the positivity conditions, we find that
\begin{align}
\Pr & ( D^{a=a_D,\overline{c}=0}_{s+1}=0 \mid \overline{L}^{a=a_D,\overline{c}=0}_{s}=\overline{l}_s,D^{a=a_D,\overline{c}=0}_{s}=0) \nonumber \\
= & \Pr ( D_{s+1}=0 \mid \overline{L}_{s}=\overline{l}_{s},D_{s}=C_{s}=0, A = a_D) ,
\label{eq: counter to obs D}
\end{align}
and an analogous argument shows that
\begin{align}
\Pr & ( L^{a = a_Y,\overline{c}=0}_{A_Y,s+1}=l_{A_Y,s+1} \mid L^{a = a_Y,\overline{c}=0}_{A_D,s+1}=l_{A_D,s+1}, \overline{L}^{a = a_Y,\overline{c}=0}_{s}=\overline{l}_s,D^{a = a_Y,\overline{c}=0}_{s+1}=0) \nonumber \\
= & \Pr ( L_{A_Y,s+1}=l_{A_Y,s+1} \mid  L_{A_D,s+1}= l_{A_D,s+1}, \overline{L}_{s}=\overline{l}_{s},D_{s+1}=C_{s}=0, A =a_Y),
\label{eq: counter to obs Ly}
\end{align}
and 
\begin{align}
\Pr & ( L^{a = a_D,\overline{c}=0}_{A_D,s+1}=l_{A_D,s+1} \mid \overline{L}^{a = a_D,\overline{c}=0}_{s}=\overline{l}_s,D^{a = a_D,\overline{c}=0}_{s+1}=0) \nonumber \\
= & \Pr ( L_{A_D,s+1}=l_{A_D,s+1} \mid \overline{L}_{s} = \overline{l}_{s},D_{s+1}=C_{s}=0, A =a_D). 
\label{eq: counter to obs Ld}
\end{align}

Similarly, 
\begin{align*}
    &\mathbb{E}  ( Y^{a_Y,a_D,\overline{c}=0} \mid 
D^{a_Y,a_D,\overline{c}=0}_ {K+1}=0,  \overline{L}^{a_Y,a_D,\overline{c}=0}_{K}=\overline{l}_{K} ) \\
=& \mathbb{E}  ( Y^{a_Y,a_D,\overline{c}=0} \mid 
D^{a_Y,a_D,\overline{c}=0}_ {K+1}=0,  \overline{L}^{a_Y,a_D,\overline{c}=0}_{K}=\overline{l}_{K}, L_0=l_0, C_0=D_0=0, A=a_Y ) \\
=& \mathbb{E}  ( Y^{a=a_Y,\overline{c}=0} \mid 
D^{a=a_Y,\overline{c}=0}_ {K+1}=0,  \overline{L}^{a=a_Y,\overline{c}=0}_{K}=\overline{l}_{K}, L_0=l_0, C_0=D_0=0, A=a_Y ) \\
= & \sum_{y} y \Pr( Y^{a=a_Y,\overline{c}=0} =y \mid D^{a=a_Y,\overline{c}=0}_ {K+1}=0,  \overline{L}^{a=a_Y,\overline{c}=0}_{K}=\overline{l}_{K}, L_0=l_0, C_0=D_0=0, A=a_Y) ,
\end{align*}
and we can then argue iteratively using exchangeability, consistency and positivity, to find that 
\begin{align*}
& \Pr( Y^{a=a_Y,\overline{c}=0} =y \mid D^{a=a_Y,\overline{c}=0}_ {K+1}=0,  \overline{L}^{a=a_Y,\overline{c}=0}_{K}=\overline{l}_{K}, L_0=l_0, C_0=D_0=0, A=a_Y)  \\ 
= & \Pr( Y =y \mid D_{k+1}=C_{k+1}=0 ,  \overline{L}_{K}=\overline{l}_{K}, A=a_Y).
\end{align*}
Hence, 
\begin{align}
 &\mathbb{E}  ( Y^{a_Y,a_D,\overline{c}=0} \mid 
D^{a_Y,a_D,\overline{c}=0}_ {K+1}=0,  \overline{L}^{a_Y,a_D,\overline{c}=0}_{K}=\overline{l}_{K} ) \nonumber \\
= & \mathbb{E}  ( Y \mid D_{K+1}=C_{K+1}=0,  \overline{L}_{K}=\overline{l}_{K}, A=a_Y).
\label{eq: counter to obs Y}
\end{align}
The proof is completed by substituting the counterfactual terms in \eqref{eq: counterfactual rep} with the observed terms \eqref{eq: counter to obs D}-\eqref{eq: counter to obs Y}, using expressions \eqref{eq: joint L,D} and \eqref{eq: marginal D}.
\end{proof}

The identification formula in Theorem \ref{theorem g formula} also follows from the general algorithm in Malinsky et al \cite{malinsky2019potential}.

\subsection{Alternative representation of the identification formula}
\label{app sec: alternative id formula proof}

\begin{theorem}
\label{theorem ipw representation}
An algebraically equivalent version of the identification formula \eqref{eq: id g-formula} in the presence of censoring is
\begin{align}
 & \frac{\mathbb{E}  [ W_{C,K}(a_Y) W_{D,K}(a_Y,a_D)  W_{L_{A_D},K}(a_Y,a_D)  (1-D_{K+1}) Y \mid A=a_Y]}{\mathbb{E}  [ W_{C,K}(a_Y) W_{D,K}(a_Y,a_D)  W_{L_{A_D},K}(a_Y,a_D)  (1-D_{K+1}) \mid A=a_Y]},
 \label{eq: alternative id formula 1}
\end{align}
where 
\begin{align*}
W_{D,K} (a_Y,a_D) &= \frac{\prod_{j=0}^{k}  \Pr(D_{j+1}=0 \mid C_{j+1}=D_{j}= 0, \overline{L}_{j},  A = a_D) }{ \prod_{j=0}^{k} \Pr(D_{j+1}=0 \mid C_{j+1}=D_{j}= 0, \overline{L}_{j},  A = a_Y) }, \\
 W_{L_{A_D},k} (a_Y,a_D) &= \frac{\prod_{j=0}^{k}   \Pr(L_{A_D,j} = l_{A_D,j} \mid C_{j}=  D_{j} = 0, \overline{L}_{j-1},  A = a_D) }{ \prod_{j=0}^{k}   \Pr(L_{A_D,j} = l_{A_D,j} \mid C_{j}= D_{j} = 0, \overline{L}_{j-1},  A = a_Y) }, \\
W_{C,K} (a_D) &= \frac{I(C_{k+1} =0) }{ \prod_{j=0}^{k}  \Pr(C_{j+1}=0 \mid C_{j}=D_{j}=0, \overline{L}_{j},  A = a_D) }.
\end{align*}
Selecting the partition $L_{A_Y,k+1}=L_{k+1}$ and $L_{A_D,k+1}=\emptyset$, consistent with identification under $A_Y$ partial isolation (see Stensrud et al. \cite[Appendix C]{stensrud2019generalized}), the weights in formula \eqref{eq: alternative id formula 1} reduce to $W_{1,K}(a_Y,a_D)$ defined in Section \ref{sec: estimation} of the main text in the absence of censoring. 
\label{theorem: weighted id formula 1}
\end{theorem}

\begin{proof}

Define
\begin{align*}
    W'_{C,k}(a) = \frac{1 }{ \prod_{j=0}^{k}  \Pr(C_{j+1}=0 \mid C_{j}=D_{j}=0, \overline{L}_{j},  A = a) }.
\end{align*}
We have
\footnotesize{
\begin{align*}
& \mathbb{E}  [ W_{C,K}(a_Y) W_{D,K}(a_Y,a_D)  W_{L_{A_D},K}(a_Y,a_D)  (1-D_{K+1}) Y \mid A=a_Y] \\ 
= & \sum_{y} \sum_{\overline{l}_k} \sum_{\overline{d}_{k+1}}\sum_{\overline{c}_{k+1}} [ \Pr(Y=y, d_{k+1},c_{k+1},\overline{L}_{k}=\overline{l}_{k} \mid A=a_Y) W'_{C,k}(a) W_{D,k}(a_Y,a_D)  W_{L_{A_D},k}(a_Y,a_D) \\
    & \times y (1-d_{k+1}) (1-c_{k+1}) ]  \\ 
=& \sum_{y} \sum_{\overline{l}_k} y  [\Pr(Y=y,\overline{L}_{k}=\overline{l}_{k}, D_{k+1}=C_{k+1}=0\mid A=a_Y) W'_{C,k}(a_Y) W_{D,k}(a_Y,a_D)  W_{L_{A_D},k}(a_Y,a_D)]  \\
=& \sum_{y}\sum_{\overline{l}_k}  y[ \Pr(Y=y \mid D_{k+1}=C_{k+1}=0,\overline{L}_{k}=\overline{l}_{k}, A=a_Y) \\
& \times \Pr( D_{k+1}=0 \mid \overline{C}_{k+1}= \overline{D}_k=0,\overline{L}_{k}=\overline{l}_{k},A=a_Y )    \\
    & \times \Pr(C_{k+1}= 0 \mid \overline{D}_k=\overline{C}_{k} =0,\overline{L}_{k}=\overline{l}_{k} , A=a_Y) \Pr(\overline{L}_{k}=\overline{l}_{k} \mid \overline{C}_{k}= \overline{D}_k=0,  A=a_Y)    \\
    & \times \Pr(\overline{D}_{k}=\overline{C}_{k}=0 \mid  A=a_Y  ) \\
    & \times W'_{C,k}(a_Y) W_{D,k}(a_Y,a_D)  W_{L_{A_D},k}(a_Y,a_D)]  \\ 
=& \sum_{y} \sum_{\overline{l}_k} y [ \Pr(Y=y \mid D_{k+1}=C_{k+1}=0,\overline{L}_{k}=\overline{l}_{k}, A=a_Y)  \\
& \times \Pr( D_{k+1}=0 \mid \overline{C}_{k+1}= \overline{D}_k=0,\overline{L}_{k}=\overline{l}_{k},A=a_Y )    \\
    & \times \Pr(C_{k+1}= 0 \mid \overline{D}_k=\overline{C}_{k} =0,\overline{L}_{k}=\overline{l}_{k}, A=a_Y) \Pr( L_k=l_{k} \mid \overline{C}_{k}= \overline{D}_k=0,\overline{L}_{k-1}=\overline{l}_{k-1},  A=a_Y)     \\
    & \times \Pr(\overline{D}_{k}=\overline{C}_{k}=0, \overline{L}_{k-1}=\overline{l}_{k-1} \mid  A=a_Y  ) \\
    & \times W'_{C,k}(a_Y) W_{D,k}(a_Y,a_D)  W_{L_{A_D},k}(a_Y,a_D)].  \\ 
\end{align*}
}
\normalsize
We use laws of probability to express $ \Pr(\overline{D}_{k}=\overline{C}_{k}=0, \overline{L}_{k-1}=\overline{l}_{k-1} \mid  A=a_Y  )$ as 
\begin{align*}
   &  \Pr(D_{k}=0 \mid C_{k}=D_{k-1}= 0, \overline{L}_{k-1}=\overline{l}_{k-1},  A = a_Y)  \nonumber \\
    & \times \Pr(C_{k}= 0 \mid D_{k-1} =  C_{k-1} =0,\overline{L}_{k-1}=\overline{l}_{k-1}, A = a_Y)  \\
   & \times \Pr(L_{k-1} = l_{k-1} \mid C_{k-1}=D_{k-1}= 0,\overline{L}_{k-2}=\overline{l}_{k-2}, A = a_Y)  \\
    & \times \Pr(\overline{D}_{k-1}=0, \overline{L}_{k-2}=\overline{l}_{k-2}, \overline{C}_{k-1}=0 \mid  A= a_Y ). \\
\end{align*}

Arguing iteratively we have
\begin{align*}
    \mathbb{E}   [ W_{C,k} & (a_Y)  W_{D,k}(a_Y,a_D)  W_{L_{A_D},k}(a_Y,a_D) Y (1-D_{k+1}) \mid A=a_Y ] \\ 
        =  \sum_{y} \sum_{\overline{l_k}} & y \Big[ \Pr(Y=y \mid  D_{k+1}=C_{k+1}=0 ,\overline{L}_{k}=\overline{l}_{k}, A=a_Y) \\
         \prod_{j=0}^{k} &  \big[ \Pr(D_{j+1}=0 \mid C_{j+1}=D_{j}= 0, \overline{L}_{j}=\overline{l}_{j},  A = a_Y)  \nonumber \\
      & \times \Pr(C_{j+1}= 0 \mid \overline{D}_j=  \overline{C}_{j} =0,\overline{L}_{j}=\overline{l}_{j},  A=a_Y)  \\
   & \times \Pr(L_{j}=l_{j} \mid C_{j}=D_{j}=0,\overline{L}_{j-1}=\overline{l}_{j-1}, A = a_Y) ]  \\
     \times  & W'_{C,k}(a_Y) W_{D,k}(a_Y,a_D)  W_{L_{A_D},k}(a_Y,a_D) \Big]  \\
        = \sum_{y} \sum_{\overline{l_k}} & y  \Big[ \Pr(Y=y \mid  D_{k+1}=C_{k+1}=0 ,\overline{L}_k=\overline{l}_k, A=a_Y) \\
        \prod_{j=0}^{k} &  \big[ \Pr(D_{j+1}=0 \mid C_{j+1}=D_{j}=0, \overline{L}_{j}=\overline{l}_{j},  A = a_Y)  \nonumber \\
       &  \times \Pr(C_{j+1}= 0 \mid \overline{D}_j= \overline{C}_{j} =0,\overline{L}_{j}=\overline{l}_{j},  A=a_Y)  \\
    &\times \Pr(L_{A_Y,j} = l_{A_Y,j} \mid C_{j}= D_{j} = 0, \overline{L}_{j-1}=\overline{l}_{j-1}, L_{A_D,j}=l_{A_D,j}, A = a_Y) \nonumber \\
    &\times \Pr(L_{A_D,j} = l_{A_D,j} \mid C_{j}= D_{j} = 0, \overline{L}_{j-1}=\overline{l}_{j-1},  A = a_Y) \big] \\
     \times & W'_{C,k}(a_Y) W_{D,k}(a_Y,a_D)  W_{L_{A_D},k}(a_Y,a_D) \Big], \\
\end{align*}
where we used that $L_k = (L_{A_Y,j},L_{A_D,j})$ in the second equality. 

By plugging in the expressions for the weights $W'_{C,k}(a_Y)$, $W_{D,k}(a_Y,a_D)$ and $  W_{L_{A_D},k}(a_Y,a_D)$ we obtain
\begin{align*}
= \sum_{y} \sum_{\overline{l_k}} & y \Big[ \Pr(Y=y \mid  D_{k+1}=C_{k+1}=0 ,\overline{L}_k=\overline{l}_k, A=a_Y) \\
    \prod_{j=0}^{k} &  \big[ \Pr(D_{j+1}=0 \mid C_{j+1}=D_{j}=0, \overline{L}_{j}=\overline{l}_{j},  A = a_D)  \nonumber \\
   &\times \Pr(L_{A_Y,j} = l_{A_Y,j} \mid C_{j}= D_{j} = 0, \overline{L}_{j-1}=\overline{l}_{j-1}, L_{A_D,j}=l_{A_D,j},
   A = a_Y) \nonumber \\
  &\times \Pr(L_{A_D,j} = l_{A_D,j} \mid C_{j}= D_{j} = 0, \overline{L}_{j-1}=\overline{l}_{j-1},  A = a_D) \big] \Big]. \\
\end{align*}

We can use analogous steps to find an expression for the denominator, 
\begin{align*}
  \mathbb{E}] & [ W_{C,K}(a_Y) W_{D,K}(a_Y,a_D)  W_{L_{A_D},K}(a_Y,a_D)  (1-D_{K+1}) \mid A=a_Y] \\
= \sum_{\overline{l_k}} &  \Big[ \prod_{j=0}^{k}   \big[ \Pr(D_{j+1}=0 \mid C_{j+1}=D_{j}=0, \overline{L}_{j}= \overline{l}_{j},  A = a_D)  \nonumber \\
   &\times \Pr(L_{A_Y,j} = l_{A_Y,j} \mid C_{j}= D_{j} = 0, \overline{L}_{j-1}=\overline{l}_{j-1}, l_{A_D,j},
   A = a_Y) \nonumber \\
  &\times \Pr(L_{A_D,j} = l_{A_D,j} \mid C_{j}= D_{j} = 0, \overline{L}_{j-1}=\overline{l}_{j-1},  A = a_D) \big] \Big], \\
\end{align*}
which completes the proof. 
\end{proof}

\section{Dismissible component conditions imply $A_Y$ partial isolation}
\label{sec: dismissible comp imply Ay partial}
\begin{lemma}
\label{lemma: diss and AY partial}
If the dismissible component conditions \eqref{ass: delta 1}-\eqref{ass: delta 3} hold, then $A_Y$ partial isolation holds. 
\end{lemma}
\begin{proof}
We give a proof by contradiction. Suppose the \eqref{ass: delta 1}-\eqref{ass: delta 3} hold, but $A_Y$ partial isolation fails. Then, if there is a direct arrow $A_Y \rightarrow D_{k+1}$ for any $k \in \{0,  \dots ,K\}$, this arrow would violate \eqref{ass: delta 2}, which is a contradiction. 

Alternatively, $A_Y$ partial isolation can only be violated if there exists a $W$ such that $A_Y \rightarrow W \rightarrow ... \rightarrow D_{k+1}$ for any $k \in \{0,  \dots ,K\}$. However, if $W \subset \overline{L}_{k}$, which implies that $W$ is measured, then \eqref{ass: delta 3} is violated, which is a contradiction. If $W$ is unmeasured, then either $\eqref{ass: delta 2}$ is violated or $\eqref{ass: delta 3}$ is violated, which is a contradiction. 
\end{proof}

\section{Non-parametric influence function and doubly robust estimators}
\label{app sec: influence doubly}
Here we derive the non-parametric influence function for the identification formula \eqref{eq: id formula no cens} from the main text. This function equals  $\mathbb{E}(Y^{a_Y,a_D} \mid D^{a_Y,a_D}_{K+1}=0)$ for $a_Y\neq a_D$ under the modified treatment assumption, $A_Y$ partial isolation, the absence of censoring and the identifying conditions of Section \ref{sec: identifiability conditions}.   


\begin{theorem}
\label{theorem: influence function}
The nonparametric influence function of the estimand \eqref{eq: id formula no cens} under a data generating model $p$ is
\begin{align}
    &  \frac{1}{\mathbb{E}_p(  1-D_{K+1} \mid    A=a_D )} \times \Bigg\{ \frac{I( A=a_D) }{P( A=a_D) } (1-D_{K+1}) \mathbb{E}_p(Y \mid D_{K+1}=0, A=a_Y, \overline{L}_{K}) \nonumber \\  
   + & \frac{I( A=a_Y)}{P( A=a_Y)} \frac{f_{\underline{L}_1,\overline{D}_{K+1} \mid L_0, A}(\underline{L}_{1},0 \mid L_0,a_D)}{f_{\underline{L}_1,\overline{D}_{K+1} \mid L_0, A}(\underline{L}_{1},0 \mid L_0,a_Y)} (1- D_{K+1}) [ Y - \mathbb{E}_p( Y \mid D_{K+1}=0,A, \overline{L}_{K})   ] \nonumber \\ 
     -  & \frac{ I(A=a_D)(1- D_{K+1})}{P( D_{K+1} = 0,    A=a_D)}  \mathbb{E}_p(Y \mid D_{K+1}=0, A=a_Y, \overline{L}_{K})  f_{\underline{L}_1,\overline{D}_{K+1} \mid L_0, A}(\underline{L}_{1},0 \mid L_0,a_D) f_{{L}_{0}}(L_{0})  \Bigg\}. \label{eq: theorem influence function}
\end{align}
\end{theorem}

\begin{proof}
We follow previous results on nonparametric influence functions \cite{van2003unified,bickel1993efficient, van2002semiparametric, tsiatis2007semiparametric}. Let $\nu^1(p_t)$ denote the nonparametric influence function for the estimand \eqref{eq: id formula no cens} under a law $p_t$, where $t \in [0,1]$ indexes a regular parametric submodel such that $p_0$ is the true data generating model. Note that we can re-express \eqref{eq: id formula no cens} as 
\begin{align*}
   \nu(p_t) & =  \frac{\alpha(p_t)}{\beta(p_t)}  =  \frac{ \mathbb{E}( \mathbb{E}(Y \mid D_{K+1}=0, \overline{L}_{K}, A=a_Y) (1-D_{K+1}) \mid    A=a_D)}{ \mathbb{E}(  1-D_{K+1} \mid    A=a_D )}.
\end{align*}
To derive the influence function for our target parameter  $\nu(p_t) = h(\alpha(p_t),\beta(p_t)) = \frac{\alpha(p_t)}{\beta(p_t)}$, we apply derivation by parts, 
\begin{align*}
    \nu^1(p_t) & = \frac{1}{\beta(p_t)} \alpha^1(p_t) - \frac{\alpha(p_t)}{[\beta(p_t)]^2}\beta^1(p_t) \\
     & = \frac{1}{\beta(p_t)} (\alpha^1(p_t) -  \nu(p_t)  \beta^1(p_t)).
\end{align*} First, 
\footnotesize{
\begin{align*}
    & \frac{d\alpha(p_t)}{dt} \Bigr\rvert_{t = 0} \\
    = &  \frac{d}{dt}  \mathbb{E}_{p_t}( \mathbb{E}_{p_0}\left(Y \mid D_{K+1}=0, A=a_Y, \overline{L}_{K}\right) (1-D_{K+1}) \mid    A=a_D )  \Bigr\rvert_{t = 0} \\
    & +  \mathbb{E}_{p_0}\left[ \frac{d}{dt} \mathbb{E}_{p_t}(Y \mid D_{K+1}=0, A=a_Y, \overline{L}_{K})  \Bigr\rvert_{t = 0}  (1-D_{K+1}) \mid    A=a_D \right]  \\
    = & \mathbb{E}_{p_0}\left[ \mathbb{E}_{p_0}(Y \mid D_{K+1}=0, A=a_Y, \overline{L}_{K}) (1-D_{K+1}) g_{D_{K+1},\overline{L}_{K} \mid A=a_D} \mid    A=a_D \right]   \\
    & +  \mathbb{E}_{p_0}\left[ \mathbb{E}_{p_0}\left[Y g_{Y \mid D_{K+1},\overline{L}_{K}, A=a_Y}  \mid D_{K+1}=0, A=a_Y, \overline{L}_{K}\right]   (1-D_{K+1}) \mid    A=a_D \right]  \\ 
    = & \mathbb{E}_{p_0}( [ \mathbb{E}_{p_0}(Y \mid D_{K+1}=0, A=a_Y, \overline{L}_{K}) (1-D_{K+1}) \\  
    & \quad  - \mathbb{E}_{p_0}(\mathbb{E}_{p_0}(Y \mid D_{K+1}=0, A=a_Y, \overline{L}_{K}) (1-D_{K+1})\mid A = a_D) ] \times g_{D_{K+1},\overline{L}_{K} \mid A=a_D} \mid    A=a_D )  \\
    & +  \mathbb{E}_{p_0}( \mathbb{E}_{p_0}( [ Y - \mathbb{E}_{p_0}(Y \mid D_{K+1}=0, A=a_Y, \overline{L}_{K})   ] g_{Y \mid D_{K+1},\overline{L}_{K}, A=a_Y}  \mid D_{K+1}=0, A=a_Y, \overline{L}_{K})  \\    & \quad \times (1-D_{K+1}) \mid    A=a_D )  \\ 
    = & \mathbb{E}_{p_0}\Bigg[ \frac{I( A=a_D)}{P( A=a_D)} \Bigg\{ \mathbb{E}_{p_0}\left[Y \mid D_{K+1}=0, A=a_Y, \overline{L}_{K}\right] (1-D_{K+1}) \\  
    & \quad  - \mathbb{E}_{p_0}\left[\frac{I( A=a_D)}{P( A=a_D)}\mathbb{E}_{p_0}(Y \mid D_{K+1}=0, A=a_Y, \overline{L}_{K}) (1-D_{K+1})\right] \Bigg\} \times g_{D_{K+1},\overline{L}_{K} \mid A}\Bigg]  \\
    & +  \mathbb{E}_{p_0}\Bigg[ \frac{I( A=a_D)}{P( A=a_D)}  \mathbb{E}_{p_0}( \frac{I( A=a_Y)}{P( A=a_Y \mid D_{K+1}=0, \overline{L}_{K})} [ Y - \mathbb{E}_{p_0}( Y \mid D_{K+1}=0,A, \overline{L}_{K})   ]  \\    
    & \quad \times g_{Y \mid D_{K+1},\overline{L}_{K}, A}  \mid D_{K+1}=0, \overline{L}_{K}) (1-D_{K+1})\Bigg]  \\
    = & \frac{1}{P( A=a_D)} \mathbb{E}_{p_0} \Bigg[ I( A=a_D)  \Bigg\{  \mathbb{E}_{p_0}(Y \mid D_{K+1}=0, A=a_Y, \overline{L}_{K}) (1-D_{K+1}) \\  
     -&  \mathbb{E}_{p_0}\left[\frac{I( A=a_D)}{P( A=a_D)}\mathbb{E}_{p_0}(Y \mid D_{K+1}=0, A=a_Y, \overline{L}_{K}) (1-D_{K+1}) \right] \Bigg\}  \underbrace{ \times (g_{Y \mid D_{K+1},\overline{L}_{K}, A}+ g_{D_{K+1},\overline{L}_{K} \mid A}+ g_{A})}_\text{Diff in exp \{ \} is mean 0 given $A$  }\Bigg]    \\
    +&   \mathbb{E}_{p_0}\Bigg[ \frac{I( A=a_D)}{P( A=a_D)}  \mathbb{E}_{p_0} \Big\{ \frac{I( A=a_Y)(1- D_{K+1})}{P( A=a_Y \mid D_{K+1}=0, \overline{L}_{K})P(D_{K+1}=0 \mid   \overline{L}_{K})} [ Y - \mathbb{E}_{p_0}( Y \mid D_{K+1}=0,A, \overline{L}_{K})   ]  \\    
    & \quad \times g_{Y \mid D_{K+1},\overline{L}_{K}, A}  \mid  \overline{L}_{K}\Big\} (1-D_{K+1})\Bigg] . \\
\end{align*}
}

\normalsize
Hereby all expected values $\mathbb{E} ( \cdot )$ are taken with respect to ${p_0}$, so we omit the subscript.  Consider the last line in the expression above,
\begin{align*}
    &   \mathbb{E}\Bigg[ \frac{I( A=a_D)}{P( A=a_D)}  \mathbb{E}\Big( \frac{I( A=a_Y)(1- D_{K+1})}{P( A=a_Y \mid D_{K+1}=0, \overline{L}_{K})P(D_{K+1}=0 \mid   \overline{L}_{K})}   \\    
    & \quad  [ Y - \mathbb{E}( Y \mid D_{K+1}=0,A, \overline{L}_{K})   ] \times g_{Y \mid D_{K+1},\overline{L}_{K}, A}  \mid  \overline{L}_{K}\Big) (1-D_{K+1})\Bigg]  \\
   = &   \mathbb{E}\Bigg[ \frac{I( A=a_D)}{P( A=a_D)} (1-D_{K+1})  \mathbb{E}\Big( \frac{I( A=a_Y)(1- D_{K+1})}{P( A=a_Y, D_{K+1}=0 \mid   \overline{L}_{K})} [ Y - \mathbb{E}( Y \mid D_{K+1}=0,A, \overline{L}_{K})   ]  \\    
    & \quad \times g_{Y \mid D_{K+1},\overline{L}_{K}, A}  \mid  \overline{L}_{K}\Big) \Bigg]  \\
   = &   \mathbb{E}\Bigg[ \frac{P( A=a_D, D_{K+1}=0 \mid \overline{L}_{K})}{P( A=a_D)} \frac{I( A=a_Y)(1- D_{K+1})}{P( A=a_Y, D_{K+1}=0 \mid   \overline{L}_{K})} [ Y - \mathbb{E}( Y \mid D_{K+1}=0,A, \overline{L}_{K})   ]  \\    
    & \quad \times g_{Y \mid D_{K+1},\overline{L}_{K}, A} \Bigg]  \\
    = &   \mathbb{E}\Bigg[ \frac{P( A=a_D, D_{K+1}=0 \mid \overline{L}_{K})}{P( A=a_D)} \frac{I( A=a_Y)(1- D_{K+1})}{P( A=a_Y, D_{K+1}=0 \mid   \overline{L}_{K})} [ Y - \mathbb{E}( Y \mid D_{K+1},A, \overline{L}_{K})   ]  \\    
    & \quad \underbrace{ \times (g_{Y \mid D_{K+1},\overline{L}_{K}, A}+ g_{D_{K+1},\overline{L}_{K}, A} )  }_\text{Expectation in [ ] is mean 0 given $D_{K+1},\overline{L}_{K}, A$} \Bigg].  \\
\end{align*}

Thus, 
\begin{align*}
& \alpha^1(p_0) \\
= & \frac{I( A=a_D)}{P( A=a_D)}    \Big[  \mathbb{E}\left(Y \mid D_{K+1}=0, A=a_Y, \overline{L}_{K}\right) (1-D_{K+1}) \\  
    & \quad - \mathbb{E} \Big( \frac{I( A=a_D)}{P( A=a_D)}   \mathbb{E}\big(Y \mid D_{K+1}=0, A=a_Y, \overline{L}_{K}\big) (1-D_{K+1}) \Big) \Big] \\
   + &\frac{P( A=a_D, D_{K+1}=0 \mid \overline{L}_{K})}{P( A=a_D)} \frac{I( A=a_Y)(1- D_{K+1})}{P( A=a_Y, D_{K+1}=0 \mid   \overline{L}_{K})} [ Y - \mathbb{E}( Y \mid D_{K+1}=0,A, \overline{L}_{K})   ] .
\end{align*}

We can derive the influence function $\beta^1(p)$ of the denominator  using a similar but simpler argument for $\frac{d\beta(p_t)}{dt} \Bigr\rvert_{t = 0}$ to find that
\begin{align*}
      \beta^1(p_0) = & \frac{ I(A=a_D)}{P(A=a_D)} [ P(D_{K+1}=1 \mid A) -D_{K+1} ] \\
       = & \frac{ I(A=a_D)}{P(A=a_D)} [  (1-D_{K+1}) -  P(D_{K+1}=0 \mid A) ].
\end{align*}

Finally, using the above results and derivation by parts we have
\footnotesize{
\begin{align*}
      \nu^1(p) =  &  \frac{1}{\beta(p)} \times \Bigg\{ \frac{I( A=a_D) }{P( A=a_D)}  \Bigg[  \mathbb{E}_p(Y \mid D_{K+1}=0, A=a_Y, \overline{L}_{K}) (1-D_{K+1}) \\  
    & \quad - \mathbb{E}_p\left[\frac{I( A=a_D)}{P( A=a_D)}\mathbb{E}_p(Y \mid D_{K+1}=0, A=a_Y, \overline{L}_{K}) (1-D_{K+1})\right] \Bigg] \\
   + &\frac{P( A=a_D, D_{K+1}=0 \mid \overline{L}_{K})}{P( A=a_D)} \frac{I( A=a_Y)(1- D_{K+1})}{P( A=a_Y, D_{K+1}=0 \mid   \overline{L}_{K})} [ Y - \mathbb{E}_p( Y \mid D_{K+1}=0,A, \overline{L}_{K})   ]  \\
  -  & \frac{ I(A=a_D)}{P(A=a_D)} [  (1-D_{K+1}) -  P(D_{K+1}=0 \mid A) ] \\
  & \times   \frac{ \mathbb{E}_p( \mathbb{E}_p(Y \mid D_{K+1}=0, A=a_Y, \overline{L}_{K}) (1-D_{K+1}) \mid    A=a_D)}{P( D_{K+1} = 0\mid    A=a_D )} \Bigg\} \\
  =  &  \frac{1}{\beta(p)} \times \Bigg\{   \frac{I( A=a_D) }{P( A=a_D)}  \Bigg[ (1-D_{K+1}) \mathbb{E}_p(Y \mid D_{K+1}=0, A=a_Y, \overline{L}_{K})  \\  
    & \quad - \mathbb{E}_p \left[\frac{I( A=a_D)}{P( A=a_D)} (1-D_{K+1}) \mathbb{E}_p(Y \mid D_{K+1}=0, A=a_Y, \overline{L}_{K}) \right] \Bigg] \\
   + &\frac{P( A=a_D, D_{K+1}=0 \mid \overline{L}_{K})}{P( A=a_D)} \frac{I( A=a_Y)(1- D_{K+1})}{P( A=a_Y, D_{K+1}=0 \mid   \overline{L}_{K})} [ Y - \mathbb{E}_p( Y \mid D_{K+1}=0,A, \overline{L}_{K})   ]  \\
  -  & \frac{ I(A=a_D)}{P(A=a_D)} \Bigg[(1-D_{K+1}) \mathbb{E}_p( \frac{I( A=a_D)(1-D_{K+1}) }{P( A=a_D)P( D_{K+1} = 0\mid    A=a_D)}  \mathbb{E}_p(Y \mid D_{K+1}=0, A=a_Y, \overline{L}_{K}) ) \\
  & -  \mathbb{E}_p\left[ \frac{I( A=a_D)(1-D_{K+1}) }{P( A=a_D)}  \mathbb{E}_p(Y \mid D_{K+1}=0, A=a_Y, \overline{L}_{K})  \right]\Bigg]\Bigg\} \\
    =  &  \frac{1}{\beta(p)} \times \Bigg\{ \frac{I( A=a_D) }{P( A=a_D)} (1-D_{K+1}) \mathbb{E}_p(Y \mid D_{K+1}=0, A=a_Y, \overline{L}_{K})  \\  
   + & \frac{I( A=a_Y)}{P( A=a_Y)} \frac{f_{\underline{L}_1,\overline{D}_{K+1} \mid L_0, A}(\underline{L}_{1},0 \mid L_0,a_D)}{f_{\underline{L}_1,\overline{D}_{K+1} \mid L_0, A}(\underline{L}_{1},0 \mid L_0,a_Y)} (1- D_{K+1}) [ Y - \mathbb{E}_p( Y \mid D_{K+1}=0,A, \overline{L}_{K})   ]  \\ 
  -  & \frac{ I(A=a_D)}{P(A=a_D)}(1- D_{K+1}) \mathbb{E}_p \left[ \frac{I( A=a_D) (1-D_{K+1})}{P( D_{K+1} = 0,    A=a_D)}  \mathbb{E}_p(Y \mid D_{K+1}=0, A=a_Y, \overline{L}_{K}) \right] \Bigg\} \\
  =  &  \frac{1}{\beta(p)} \times \Bigg\{ \frac{I( A=a_D) }{P( A=a_D) } (1-D_{K+1}) \mathbb{E}_p(Y \mid D_{K+1}=0, A=a_Y, \overline{L}_{K})  \\  
   + & \frac{I( A=a_Y)}{P( A=a_Y)} \frac{f_{\underline{L}_1,\overline{D}_{K+1} \mid L_0, A}(\underline{L}_{1},0 \mid L_0,a_D)}{f_{\underline{L}_1,\overline{D}_{K+1} \mid L_0, A}(\underline{L}_{1},0 \mid L_0,a_Y)} (1- D_{K+1}) [ Y - \mathbb{E}_p( Y \mid D_{K+1}=0,A, \overline{L}_{K})   ]  \\ 
     -  & \frac{ I(A=a_D)(1- D_{K+1})}{P( D_{K+1} = 0,    A=a_D)}  \mathbb{E}_p(Y \mid D_{K+1}=0, A=a_Y, \overline{L}_{K})  f_{\underline{L}_1,\overline{D}_{K+1} \mid L_0, A}(\underline{L}_{1},0 \mid L_0,a_D) f_{{L}_{0}}(L_{0})  \Bigg\} ,
\end{align*}
}

\normalsize
where the last line in the last equality follows from iterative expectations because
\begin{align*}
   & \mathbb{E}_p \left[ \frac{I( A=a_D) (1-D_{K+1})}{P( D_{K+1} = 0,    A=a_D)} \mid \overline{L}_{K} \right] \\
  =   & \frac{\mathbb{E}_p [  I( A=a_D) (1-D_{K+1})  \mid \overline{L}_{K} ]  }{P( A=a_D)P( D_{K+1} = 0\mid    A=a_D)}\\
  = &  \frac{f_{\underline{L}_1,\overline{D}_{K+1} \mid L_0, A}(\underline{L}_{1},0 \mid L_0,a_D) }{P( D_{K+1} = 0 \mid    A=a_D) f_{\underline{L}_{1}}(\underline{L}_{1})}. \\
\end{align*}
\end{proof}

\subsection{Constructing a doubly robust estimator}
\label{sec: app doubly rob}
Let superscript $\sim$ denote an estimator. The influence function \eqref{eq: theorem influence function} motivates the following M-estimator $\hat{\nu}_{dr,a_Y,a_D}$, which is a solution to
\begin{align*}
    \sum_{i=1}^{n}\hat{U}_{i}(\nu_{a_Y,a_D})=0,
\end{align*}
where
\footnotesize{
\begin{align*}
    U(\nu_{a_Y,a_D}) & =  \frac{1}{\beta(p)} \times \Bigg\{ \frac{I( A=a_D) }{P( A=a_D)} (1-D_{K+1}) \mathbb{E}(Y \mid D_{K+1}=0, A=a_Y, \overline{L}_{K})  \\  
   + & \frac{I( A=a_Y)}{P( A=a_Y)} \frac{f_{\underline{L}_1,\overline{D}_{K+1} \mid L_0, A}(\underline{L}_{1},0 \mid L_0,a_D)}{f_{\underline{L}_1,\overline{D}_{K+1} \mid L_0, A}(\underline{L}_{1},0 \mid L_0,a_Y)} (1- D_{K+1}) [ Y - \mathbb{E}( Y \mid D_{K+1}=0,A, \overline{L}_{K})   ]   \Bigg\} - \nu_{a_Y,a_D}, 
\end{align*}
}
and, analogously, $\hat{U}(\nu_{a_Y,a_D})$ is defined as above but evaluated under the estimators $\Tilde{\mathbb{E}}(Y \mid D_{K+1}=0, A, \overline{L}_{K})$, $\Tilde{f}_{\underline{L}_1,D_{K+1} \mid L_0, A}(\underline{L}_{1},0 \mid L_0, A)$ and $\Tilde{\beta}(p_t)$. 

Note that the estimator $\hat{\nu}_{dr,a_Y,a_D}$ is equivalent to the one step estimator derived from the influence function $\nu^1_{\Tilde{P}}$, specifically $\hat{\nu}_{dr,a_Y,a_D} = \nu (\Tilde{P}) + \mathbb{P}_n ( \nu^1_{\Tilde{P}}) $ where $\nu (\Tilde{P})$ is the naive plug-in estimator of $\nu_{a_Y,a_D}$ and $ \mathbb{P}_n ( \cdot)$ is the empirical law. This can be seen by the following argument: let $\theta_{\Tilde{P}}$ be the last term of $\nu^1_{\Tilde{P}}$ in \eqref{eq: theorem influence function}, and note that $\mathbb{P}_n ( \theta_{\Tilde{P}} ) = \nu (\Tilde{P})$, that is,
\footnotesize{
\begin{align*}
    & \mathbb{P}_n(    \theta_{\Tilde{P}} )  \\
    = & \mathbb{P}_n\left( \frac{ I(A=a_D)(1- D_{K+1})}{\Tilde{\beta} \Tilde{P}( D_{K+1} = 0,    A=a_D)}  \Tilde{\mathbb{E}}(Y \mid D_{K+1}=0, A=a_Y, \overline{L}_{K})  \tilde{P}_{\underline{L}_1,D_{K+1} \mid L_0, A}(\underline{L}_{1},0 \mid L_0,a_D) \tilde{f}_{{L}_{0}}(L_{0})\right) \\
    = & \mathbb{P}_n\left( \frac{1}{\Tilde{\beta}} \Tilde{\mathbb{E}}(Y \mid D_{K+1}=0, A=a_Y, \overline{L}_{K})  \tilde{P}_{\underline{L}_1,D_{K+1} \mid L_0, A}(\underline{L}_{1},0 \mid L_0,a_D) \tilde{f}_{{L}_{0}}(L_{0})\right), 
\end{align*}
}

\normalsize
where $\Tilde{\beta}(p)$ are $\Tilde{P}( D_{K+1} = 0,    A=a_D)$  are empirical averages. 

Given that $\beta(p)$ can be estimated consistently with the sample average $\Tilde{\beta} = \mathbb{E}_n(1-D_{K+1} \mid A=a_D)$, in the next section we show that $\hat{\nu}_{dr,a_Y,a_D}$ is doubly robust,  in the sense that $\mathbb{E}( \hat{U}(\nu_{a_Y,a_D})) = 0$ if either $\mathbb{E}(Y \mid D_{K+1}=0, A, \overline{L}_{K})$ or $f_{\underline{L}_1,\overline{D}_{K+1} \mid L_0, A}(\underline{L}_{1},0 \mid L_0,a)$ is consistently estimated  by $\tilde{\mathbb{E}}(Y \mid D_{K+1}=0, A, \overline{L}_{K})$ or $\tilde{f}_{\underline{L}_1,D_{K+1} \mid L_0, A}(\underline{L}_{1},0 \mid L_0,a)$ but not necessarily both. We show this in Section \ref{sec: proof doubly robust}.

\subsection{Proof of doubly robustness}
\label{sec: proof doubly robust}
First we express
$\mathbb{E}( \hat{U}(\nu_{a_Y,a_D})  + \nu_{a_Y,a_D}) $

\footnotesize{
\begin{align*}
 &\mathbb{E}( \hat{U}(\nu_{a_Y,a_D}) + \nu_{a_Y,a_D} ) \\
=  &  \frac{1}{\Tilde{\beta}}  \times \Bigg\{ \mathbb{E} [ \frac{I( A=a_D) }{P( A=a_D)} (1-D_{K+1}) \Tilde{\mathbb{E}}(Y \mid D_{K+1}=0, A=a_Y, \overline{L}_{K}) ] \\  %
   + & \mathbb{E} [ \frac{I( A=a_Y)}{P( A=a_Y)} \frac{\Tilde{f}_{\underline{L}_1,D_{K+1} \mid L_0, A}(\underline{L}_{1},0 \mid L_0,a_D)}{\Tilde{f}_{\underline{L}_1,D_{K+1} \mid L_0, A}(\underline{L}_{1},0 \mid L_0,a_Y)} (1- D_{K+1}) [ Y - \Tilde{\mathbb{E}}( Y \mid D_{K+1}=0,A, \overline{L}_{K}) ]  ]   \Bigg\}. \\ %
  =  &  \frac{1}{\Tilde{\beta}}  \times \Bigg\{ \mathbb{E} [\mathbb{E} [ \frac{I( A=a_D) }{P( A=a_D)} (1-D_{K+1}) \Tilde{\mathbb{E}}(Y \mid D_{K+1}=0, A=a_Y, \overline{L}_{K}) \mid \overline{L}_{K} ] ] \\  
   + & \mathbb{E} [ \mathbb{E} [ \frac{I( A=a_Y)}{P( A=a_Y)} \frac{\Tilde{f}_{\underline{L}_1,D_{K+1} \mid L_0, A}(\underline{L}_{1},0 \mid L_0,a_D)}{\Tilde{f}_{\underline{L}_1,D_{K+1} \mid L_0, A}(\underline{L}_{1},0 \mid L_0,a_Y)} (1- D_{K+1}) [ Y - \Tilde{\mathbb{E}}( Y \mid D_{K+1}=0,A, \overline{L}_{K}) ] \mid \overline{L}_{K} ] ]  ]   \Bigg\}. \\ %
=  &  \frac{1}{\Tilde{\beta}}  \times \Bigg\{ \mathbb{E} [ \Tilde{\mathbb{E}}(Y \mid D_{K+1}=0, A=a_Y, \overline{L}_{K}) \mathbb{E}  [ \frac{I( A=a_D) }{P( A=a_D)} (1-D_{K+1}) \mid \overline{L}_{K} ] ] \\  
   + & \mathbb{E} \Big[ \frac{1}{P( A=a_Y)} \frac{\Tilde{f}_{\underline{L}_1,D_{K+1} \mid L_0, A}(\underline{L}_{1},0 \mid L_0,a_D)}{\Tilde{f}_{\underline{L}_1,D_{K+1} \mid L_0, A}(\underline{L}_{1},0 \mid L_0,a_Y)} \{ \mathbb{E} [ I( A=a_Y) (1- D_{K+1}) Y \mid  \overline{L}_{K}] \\
   & \quad - \mathbb{E} [ I( A=a_Y) (1- D_{K+1}) \Tilde{\mathbb{E}}( Y \mid D_{K+1}=0,A, \overline{L}_{K})  \mid  \overline{L}_{K} ] \} \Big]     \Bigg\}. \\ 
   =  &  \frac{1}{\Tilde{\beta}}  \times \Bigg\{ \mathbb{E} [ \Tilde{\mathbb{E}}(Y \mid D_{K+1}=0, A=a_Y, \overline{L}_{K}) \frac{1}{P( A=a_D \mid L_0 )} \mathbb{P}  (D_{K+1}=0,A=a_D \mid \overline{L}_{K} ) ] \\  
   + & \mathbb{E} \Big[ \frac{1}{P( A=a_Y)} \frac{\Tilde{f}_{\underline{L}_1,D_{K+1} \mid L_0, A}(\underline{L}_{1},0 \mid L_0,a_D)}{\Tilde{f}_{\underline{L}_1,D_{K+1} \mid L_0, A}(\underline{L}_{1},0 \mid L_0,a_Y)} \{ \mathbb{E} [ I( A=a_Y) (1- D_{K+1}) \mathbb{E}( Y \mid D_{K+1}=0,A, \overline{L}_{K}) \mid  \overline{L}_{K}] \\ 
   & \quad - \mathbb{E} [ I( A=a_Y) (1- D_{K+1}) \Tilde{\mathbb{E}}( Y \mid D_{K+1}=0,A, \overline{L}_{K})  \mid  \overline{L}_{K} ] \} \Big]     \Bigg\}. \\ 
   =  &  \frac{1}{\Tilde{\beta}}  \times \Bigg\{ \mathbb{E} \left[ \Tilde{\mathbb{E}}(Y \mid D_{K+1}=0, A=a_Y, \overline{L}_{K}) \frac{f_{\underline{L}_1,\overline{D}_{K+1} \mid L_0, A}(\underline{L}_{1},0 \mid L_0,a_D) }{ f_{\underline{L}_{1}}(\underline{L}_{1})} \right] \\  
   + & \mathbb{E} \Bigg[ \frac{\Tilde{f}_{\underline{L}_1,D_{K+1} \mid L_0, A}(\underline{L}_{1},0 \mid L_0,a_D)}{\Tilde{f}_{\underline{L}_1,D_{K+1} \mid L_0, A}(\underline{L}_{1},0 \mid L_0,a_Y)} \\
   & \quad \mathbb{E}\Big[ \frac{I( A=a_Y)}{P( A=a_Y)}  (1- D_{K+1}) \{  \mathbb{E}( Y \mid D_{K+1}=0,A, \overline{L}_{K})-\Tilde{\mathbb{E}}( Y \mid D_{K+1}=0,A, \overline{L}_{K}) \} \mid  \overline{L}_{K} \Big]  \Bigg]     \Bigg\}.
   \\ 
\end{align*}
} 

\normalsize
From the final expression above, we see that the estimator is doubly robust. First, if $\Tilde{\mathbb{E}}( Y \mid D_{K+1}=0,A, \overline{L}_{K})$ is correctly specified but $\Tilde{f}_{\underline{L}_1,D_{K+1} \mid L_0, A}(\underline{L}_{1},0 \mid L_0,a_D)$ is not neccessarily correctly specified, then the second line in the final equality above is 0 and the first line is the estimand of interest, that is,
\footnotesize{
\begin{align*}
   =  &  \frac{1}{\Tilde{\beta}}  \times \Bigg\{ \mathbb{E} \left[ \Tilde{\mathbb{E}}(Y \mid D_{K+1}=0, A=a_Y, \overline{L}_{K}) \frac{f_{\underline{L}_1,\overline{D}_{K+1} \mid L_0, A}(\underline{L}_{1},0 \mid L_0,a_D) }{ f_{\underline{L}_{1}}(\underline{L}_{1})} \right] \\  
   + & \mathbb{E} \Big[ \frac{\Tilde{f}_{\underline{L}_1,D_{K+1} \mid L_0, A}(\underline{L}_{1},0 \mid L_0,a_D)}{\Tilde{f}_{\underline{L}_1,D_{K+1} \mid L_0, A}(\underline{L}_{1},0 \mid L_0,a_Y)} \\
   & \qquad \times  \mathbb{E} [ \frac{I( A=a_Y)}{P( A=a_Y)}  (1- D_{K+1}) \underbrace{ \{  \mathbb{E}( Y \mid D_{K+1}=0,A, \overline{L}_{K})-\Tilde{\mathbb{E}}( Y \mid D_{K+1}=0,A, \overline{L}_{K}) \}}_{=0 \  \forall \overline{l}_{k}  \in \mathcal{\overline{L} }_{k}}  \mid  \overline{L}_{K} ]  \Big]     \Bigg\}. \\ 
    =  &  \frac{1}{\Tilde{\beta}}  \times \sum_{\overline{l}_{k}} \left[ \Tilde{\mathbb{E}}(Y \mid D_{K+1}=0, A=a_Y, \overline{L}_{K}=\overline{l}_{K}) f_{\overline{L}_K,D_{K+1} \mid A}( \overline{l}_{k},0 \mid a_D) \right].\\  
\end{align*}
}
\normalsize
Second, if $\Tilde{f}_{\underline{L}_1,D_{K+1} \mid L_0, A}(\underline{L}_{1},0 \mid L_0,a_D)$ is correctly specified but not necessarily $\Tilde{\mathbb{E}}( Y \mid D_{K+1}=0,A, \overline{L}_{K})$ is correctly specified, then 
\begin{align*}
   =  &  \frac{1}{\Tilde{\beta}}  \times \Bigg\{ \mathbb{E} \Big[ \Tilde{\mathbb{E}}(Y \mid D_{K+1}=0, A=a_Y, \overline{L}_{K}) \frac{f_{\underline{L}_1,\overline{D}_{K+1} \mid L_0, A}(\underline{L}_{1},0 \mid L_0,a_D) }{ f_{\underline{L}_{1}}(\underline{L}_{1})} \Big] \\  
   + & \mathbb{E} \Big[ \frac{\Tilde{f}_{\underline{L}_1,D_{K+1} \mid L_0, A}(\underline{L}_{1},0 \mid L_0,a_D)}{\Tilde{f}_{\underline{L}_1,D_{K+1} \mid L_0, A}(\underline{L}_{1},0 \mid L_0,a_Y)}  \mathbb{E} [ \frac{I( A=a_Y)}{P( A=a_Y \mid L_0)} \\
   & \qquad \times (1- D_{K+1}) \{  \mathbb{E}( Y \mid D_{K+1}=0,A, \overline{L}_{K})-\Tilde{\mathbb{E}}( Y \mid D_{K+1}=0,A, \overline{L}_{K}) \} \mid  \overline{L}_{K} ]  \Big]     \Bigg\}. \\ 
   =  &  \frac{1}{\Tilde{\beta}} \times \Bigg\{   
   \sum_{\overline{l}_{k}} \left[ \Tilde{\mathbb{E}}(Y \mid D_{K+1}=0, A=a_Y, \overline{L}_{K}=\overline{l}_{k}) f_{\overline{L}_K,D_{K+1} \mid A}( \overline{l}_{K},0 \mid a_D) \right] \\  
      + & \mathbb{E} \left[ \frac{\Tilde{f}_{\underline{L}_1,D_{K+1} \mid L_0, A}(\underline{L}_{1},0 \mid L_0,a_D)}{\Tilde{f}_{\underline{L}_1,D_{K+1} \mid L_0, A}(\underline{L}_{1},0 \mid L_0,a_Y)}   \frac{f_{\underline{L}_1,\overline{D}_{K+1} \mid L_0, A}(\underline{L}_{1},0 \mid L_0,a_Y) }{ f_{\underline{L}_{1}}(\underline{L}_{1})} \mathbb{E}( Y \mid D_{K+1}=0,A, \overline{L}_{K})\mid  \overline{L}_{K}   \right]     . \\ 
     & - \mathbb{E} \left[ \frac{\Tilde{f}_{\underline{L}_1,D_{K+1} \mid L_0, A}(\underline{L}_{1},0 \mid L_0,a_D)}{\Tilde{f}_{\underline{L}_1,D_{K+1} \mid L_0, A}(\underline{L}_{1},0 \mid L_0,a_Y)}  \frac{f_{\underline{L}_1,\overline{D}_{K+1} \mid L_0, A}(\underline{L}_{1},0 \mid L_0,a_Y) }{ f_{\underline{L}_{1}}(\underline{L}_{1})} \Tilde{\mathbb{E}}( Y \mid D_{K+1}=0,A, \overline{L}_{K})   \right]     \Bigg\}. \\ 
   =  &  \frac{1}{\Tilde{\beta}}  \times  \Bigg\{  
   \sum_{\overline{l}_{K}} \left[ \Tilde{\mathbb{E}}(Y \mid D_{K+1}=0, A=a_Y, \overline{L}_{K}=\overline{l}_{k}) f_{\overline{L}_K,D_{K+1} \mid A}( \overline{l}_{k},0 \mid a_D) \right] \\  
      + & \sum_{\overline{l}_{k}} \left[ \mathbb{E}(Y \mid D_{K+1}=0, A=a_Y, \overline{L}_{K}=\overline{l}_{k}) \Tilde{f}_{\underline{L}_1,D_{K+1} \mid L_0, A}(\underline{L}_{1},0 \mid L_0, a_D) \right] \\
      & - \sum_{\overline{l}_{k}} \left[ \Tilde{\mathbb{E}}(Y \mid D_{K+1}=0, A=a_Y, \overline{L}_{K}=\overline{l}_{K}) \Tilde{f}_{\underline{L}_1,D_{K+1} \mid L_0, A}(\underline{L}_{1},0 \mid L_0, a_D) \right]    \Bigg\} \\ 
      = &  \frac{1}{\Tilde{\beta}} \sum_{\overline{l}_{k}} \left[ \mathbb{E}(Y \mid D_{K+1}=0, A=a_Y, \overline{L}_{K}=\overline{l}_{K}) \Tilde{f}_{\underline{L}_1,D_{K+1} \mid L_0, A}(\underline{L}_{1},0 \mid L_0, a_D) \right] .
\end{align*} 

\normalsize
\section{Censoring in the data analysis of Section \ref{sec: data example}}
\label{sec: appendix censoring}
In our illustrative data example in Section \ref{sec: data example}, there were no missing values for $(D_{K+1}, \overline{L}_K, A)$, $K=11$, but some subjects had missing values for $Y$; that is, quality of life at 12 months of follow-up was unknown in some individuals. This coincides with the censored data structure of Appendix A for the special case where $\overline{C}_{11} = 0$ and a modified temporal order assumption in interval $K+1$ such that $D_{K+1}$ precedes $C\equiv C_{K+1}$. In this case, the identifying function \eqref{eq: id g-formula} reduces to
\begin{align}
   \nu(p) & =   \frac{ \mathbb{E}_p\left( \mathbb{E}_p\left(\frac{Y(1-C)}{ \Pr(C=0 \mid D_{K+1}=0, \overline{L}_{K},  A = a_Y) } \mid D_{K+1}=0, \overline{L}_{K}, A=a_Y\right) (1-D_{K+1}) \mid    A=a_D\right)}{ \mathbb{E}_p(  1-D_{K+1} \mid    A=a_D )},
   \label{eq: cens id for ex}
\end{align}
consistent with identification under $A_Y$ partial isolation.

To fit the outcome regression estimator $\hat{\nu}_{or,a_Y,a_D} $, we used the estimator described in Section \ref{sec: estimation}, but the outcome regression was restricted to the uncensored observations. 

To fit the weighted estimator $\hat{\nu}_{ipw,a_Y,a_D} $ in the application in Section \ref{sec: data example} the main text, we adjusted for censoring using the censoring weights $W_{C,K}$ from Appendix \ref{app sec: alternative id formula proof} based on the following model for the weight denominator:
\begin{align}
& \text{logit} [\Pr(C=1 \mid D_{12}=0,A, L_0,L_k; \alpha_C)] \nonumber \\
 & = \alpha_{C,0} + \alpha_{C,1}A +\alpha'_{C,2}L_0+\alpha_{C,3}L_{11}. \label{eq: censoring weights} 
\end{align}

In this case, we used a modified doubly robust estimator $\hat{\nu}_{dr,a_Y,a_D}$  based on the nonparametric influence function for \eqref{eq: cens id for ex},
\footnotesize{
\begin{align}
    &  \frac{1}{\mathbb{E}_p(  1-D_{K+1} \mid    A=a_D )} \times \Bigg\{ \frac{I( A=a_D) }{P( A=a_D) } (1-D_{K+1}) \mathbb{E}_p(Y \mid D_{K+1}=C=0, A=a_Y, \overline{L}_{K}) \nonumber \\  
   + & \frac{I( A=a_Y)}{P( A=a_Y)} \frac{f_{\underline{L}_1,\overline{D}_{K+1} \mid L_0, A}(\underline{L}_{1},0 \mid L_0,a_D)}{f_{\underline{L}_1,\overline{D}_{K+1} \mid L_0, A}(\underline{L}_{1},0 \mid L_0,a_Y)}(1-D_{K+1})(1-C) \frac{[ Y - \mathbb{E}_p( Y \mid D_{K+1}=C=0,A, \overline{L}_{K})   ]}{\Pr(C=0 \mid D_{K+1}=0, \overline{L}_{K},  A = a_Y)}   \nonumber \\ 
     -  & \frac{ I(A=a_D)(1- D_{K+1})}{P( D_{K+1} = 0,    A=a_D)}  \mathbb{E}_p(Y \mid D_{K+1}=C=0, A=a_Y, \overline{L}_{K})  f_{\underline{L}_1,\overline{D}_{K+1} \mid L_0, A}(\underline{L}_{1},0 \mid L_0,a_D) f_{{L}_{0}}(L_{0})  \Bigg\}, \nonumber
\end{align}
}

\normalsize
 which is derived using straightforward extensions of the arguments in Appendix \ref{app sec: influence doubly}. This motivates the estimating function
\begin{align*}
    U(\nu_{a_Y,a_D}) & =  \frac{1}{\beta(p)} \times \Bigg\{ \frac{I( A=a_D) }{P( A=a_D)} (1-D_{K+1}) \mathbb{E}(Y \mid D_{K+1}=C=0, A=a_Y, \overline{L}_{K})  \\  
   + & \frac{I( A=a_Y)}{P( A=a_Y)} \frac{f_{\underline{L}_1,\overline{D}_{K+1} \mid L_0, A}(\underline{L}_{1},0 \mid L_0,a_D)}{f_{\underline{L}_1,\overline{D}_{K+1} \mid L_0, A}(\underline{L}_{1},0 \mid L_0,a_Y)} (1- D_{K+1}) (1-C) \\
  & \quad \times \frac{[ Y - \mathbb{E}_p( Y \mid D_{K+1}=C=0,A, \overline{L}_{K})   ]}{\Pr(C=0 \mid D_{K+1}=0, \overline{L}_{K},  A = a_Y)}   \Bigg\} - \nu_{a_Y,a_D}, 
\end{align*}
and, similar to Section \ref{sec: app doubly rob}, we define $\hat{U}(\nu_{a_Y,a_D})$ as above but evaluated under the estimators $\Tilde{\mathbb{E}}(Y \mid D_{K+1}=0, C=0, A, \overline{L}_{K})$, $\Tilde{f}_{\underline{L}_1,D_{K+1} \mid L_0, A}(\underline{L}_{1},0 \mid L_0, A)$, $\Tilde{\Pr}(C=0 \mid D_{K+1}=0, \overline{L}_{K},  A = a_Y)$ and $\Tilde{\beta}(p_t)$.  Let $\hat{\nu}_{dr,a_Y,a_D}$ be the solution to the estimating equation
$ \sum_{i=1}^{n}\hat{U}_{i}(\nu_{a_Y,a_D})=0$. This
estimator is consistent if $\Tilde{\mathbb{E}}( Y \mid D_{K+1}=0,C=0, A, \overline{L}_{K})$  is correctly specified, or if both $\Tilde{f}_{\underline{L}_1,D_{K+1} \mid L_0, A}(\underline{L}_{1},0 \mid L_0,a_D)$ and $\Tilde{\Pr}(C=0 \mid D_{K+1}=0, \overline{L}_{K},  A = a_Y)$ are correctly specified.  We computed this estimator under the model \eqref{eq: censoring weights} above for censoring and the models specified in Section \ref{sec: data example}, but with the outcome model restricted to uncensored individuals.
\end{document}